\documentclass[a4paper,11pt,fleqn]{article} 

\usepackage{amsmath,amssymb,amsthm} 
\usepackage[mathscr]{eucal}

\usepackage{hyperref}
\hypersetup{colorlinks, linkcolor=blue, citecolor=blue, urlcolor=blue}

\allowdisplaybreaks

\makeatletter
\long\def\@makecaption#1#2{%
  \vskip\abovecaptionskip\footnotesize
  \sbox\@tempboxa{#1. #2}%
  \ifdim \wd\@tempboxa >\hsize
    #1. #2\par
  \else
    \global \@minipagefalse
    \hb@xt@\hsize{\hfil\box\@tempboxa\hfil}%
  \fi
  \vskip\belowcaptionskip}
\makeatother

\newcommand{\todo}[1][\null]{\ensuremath{\clubsuit}}

\newcommand{\noprint}[1]{}
\newcommand{\checked}[1][\null]{\ensuremath{\boldsymbol{\surd}}}

\newcommand{\p}{\partial}

\newcommand{\sgn}{\mathop{\rm sgn}\nolimits}
\newcommand{\EqOrd}{r}

\newtheorem{theorem}{Theorem}

\newtheorem{corollary}[theorem]{Corollary}
\newtheorem{proposition}[theorem]{Proposition}
\newtheorem*{problem*}{Problem}
{\theoremstyle{definition}

\newtheorem{remark}[theorem]{Remark}
\newtheorem*{remark*}{Remark}
}

\begin{document}

\par\noindent {\LARGE\bf
Equivalence groupoid and group classification\\ of a class of variable-coefficient Burgers equations
\par}

\vspace{4mm}\par\noindent {\large 
Stanislav Opanasenko$^{\dag\ddag}$, Alexander Bihlo$^\dag$ and Roman O.\ Popovych$^{\ddag\S}$,
\par}

\vspace{4mm}\par\noindent{\it
$^{\dag}$Department of Mathematics and Statistics, Memorial University of Newfoundland,\\
$\phantom{^{\dag}}$~St.\ John's (NL) A1C 5S7, Canada
}

\vspace{2mm}\par\noindent {\it
$^\ddag$Institute of Mathematics of NAS of Ukraine, 3 Tereshchenkivska Str., 01024 Kyiv, Ukraine\par
}

\vspace{2mm}\par\noindent {\it
$^{\S}$Fakult\"at f\"ur Mathematik, Universit\"at Wien, Oskar-Morgenstern-Platz 1, A-1090 Wien, Austria
}

\vspace{2mm}\par\noindent {\it
\textup{E-mail:} sopanasenko@mun.ca, abihlo@mun.ca, rop@imath.kiev.ua
}\par

\vspace{8mm}\par\noindent\hspace*{5mm}\parbox{150mm}{\small
We study admissible transformations and Lie symmetries for a class of variable-coefficient Burgers equations.
We combine the advanced methods of splitting into normalized subclasses and
of mappings between classes that are generated by families of point transformations
parameterized by arbitrary elements of the original classes.
A nontrivial differential constraint on the arbitrary elements of the class of variable-coefficient Burgers equations
leads to its partition into two subclasses, which are related to normalized classes
via families of point transformations parameterized by subclasses' arbitrary elements.
One of the mapped classes is proved to be normalized in the extended generalized sense,
and its effective extended generalized equivalence group is found.
Using the mappings between classes and the algebraic method of group classification,
we carry out the group classification of the initial class with respect to its equivalence groupoid.
}\par\vspace{4mm}

\noprint{
35B06  	Symmetries, invariants, etc.
35K59  	Quasilinear parabolic equations
}

\section{Introduction}

Suggested initially as a model for one-dimensional turbulence, the Burgers equation and its generalizations
are nowadays used for describing various phenomena in physics, in particular, in acoustics, condensed matter and statistical physics, 
as well as for studying non-physics problems such as vehicular traffic,
see, e.g., \cite{ChowdhurySantenSchadschneider2000,Crighton1979,Sachdev2009} and references therein.

The aim of the present paper is to study, from the point of view of symmetry analysis,
the class~$\mathcal L$ of variable-coefficient Burgers equations of the general form
\begin{gather}\label{Qu:QuClass}
u_t+C(t,x)uu_x=A^2(t,x)u_{xx},
\end{gather}
where $A^2$ and~$C$ are smooth functions of~$(t,x)$ with~$A^2C\neq0$.

\looseness=-1
The investigation of admissible transformations and Lie symmetries of Burgers-like equations
dates back to~\cite{katk1965a}, where the maximal Lie invariance algebra of the Burgers equation
was computed in the course of the group classification of diffusion--convection equations of the general form $u_t+uu_x=(f(u)u_x)_x$.
Conformal transformations between equations from the class~$\mathcal L$ with $C=1$ and $A^2_x=0$ were found in~\cite{Cates1989}.
Lie symmetries and similarity solutions of such equations were considered in~\cite{DoyleEnglefield1990} (see also~\cite{poch2014a,Soh2004}).
Later, Kingston and Sophocleous in~\cite{KingstonSophocleous1991} extended this consideration,
having implicitly computed the equivalence groupoid of the wider subclass~$\mathcal L_{0'}$ of~$\mathcal L$ associated with the constraint~$C=1$.
In fact, this was the first computation of the equivalence groupoid for a class of differential equations in the literature.
Moreover, the subclass~$\mathcal L_{0'}$ happens to be normalized as was shown in~\cite{PocheketaPopovych2017}.
The extended symmetry analysis thereof was also carried out in~\cite{PocheketaPopovych2017},
which included, in particular, solving the group classification problem for this subclass
and the construction of exact solutions to equations therein.
The paper~\cite{PocheketaPopovych2017} contains an extended review of results 
related to the symmetry analysis for equations in the class~$\mathcal L_{0'}$ as well.

In~\cite{GazeauWinternitz1992,WinternitzGazeau1992},
admissible transformations (which are called `allowed transformations' there)
of the class~$\mathcal K$ of variable-coefficient Korteweg--de Vries equations
$u_t+C(t,x)uu_x=A^3(t,x)u_{xxx}$ analogous to~\eqref{Qu:QuClass},
where $A^3C\neq0$, were computed and used for the classification of Lie symmetries of such equations.
See~\cite{VaneevaPosta2017} for a modern interpretation of these results
and~\cite{GagnonWinternitz1993} for the application of the same technique
to a class of (1+1)-dimensional variable-coefficient Schr\"odinger equations.
An attempt at carrying over the results of~\cite{GazeauWinternitz1992,WinternitzGazeau1992} to the class~$\mathcal L$
was made in~\cite{Qu1995} under the assumption
that the equivalence groupoids of the classes~$\mathcal K$ and~$\mathcal L$ are of the same structure.
In fact, this is not the case
as the $x$-components of admissible transformations in the class~$\mathcal K$
are necessarily affine with respect to~$x$, unlike those in the class~$\mathcal L$.
The main complication of studying the class~$\mathcal L$ in comparison with the class~$\mathcal K$ is
that the latter class splits into subclasses that are normalized at least in the extended generalized sense~\cite{VaneevaPosta2017},
while this is not the case for the class~$\mathcal L$.
To partition the class~$\mathcal L$ into more suitable subclasses, which are nevertheless not normalized, 
we need to overcome the problem by studying first the class~$\hat{\mathcal L}$ 
that is constituted by the equations of the form 
\begin{gather}\label{Qu:QuClassMapped}
u_t+uu_x=\hat A^2(t,x)u_{xx}+\hat A^1(t,x)u_x
\end{gather}
with~$\hat A^2\neq0$ and which is weakly similar to~$\mathcal L$ and then to return to the original class.
In fact, the desired partition for~$\mathcal L$ is achievable only
by virtue of a nontrivial differential condition on the arbitrary elements~$A^2$ and~$C$.
It also turns out that one of the subclasses of~$\hat{\mathcal L}$ is normalized in the extended generalized sense.
Furthermore, the normalization of the superclass~$\mathcal B$ of~$\mathcal L$ 
that consists of the general Burgers equations~\cite{OpanasenkoBihloPopovych2017} is of partial use here
since there is no mapping of the class~$\mathcal B$ to the class~$\mathcal L$
via gauging of arbitrary elements of the class~$\mathcal B$ by its equivalence transformations.

\looseness=-1
The classes~$\mathcal L$ and~$\hat{\mathcal L}$ are rich in equations of physical importance 
\cite{Crighton1979,Sachdev2009}.
At the same time, the independent and dependent variables are usually related to physical
variables and coordinates in a rather complicated manner for these equations. 
The roles of~$t$ and~$x$ are often interchanged, i.e., 
the variables~$t$ and~$x$ are interpreted as the space- and time-type values like the range and the retarded time, respectively. 
The sign of the nonlinearity is alternated as well, which merely gives an equivalent form of Burgers-like equations 
modulo alternating the sign of~$u$.
An important class of generalized Burgers equations consists of the so-called nonplanar Burgers equations of the form  
\begin{gather}\label{eq:NonplanarBurgersEq}
u_t+uu_x+\frac{A_t(t)}{2A(t)}u=u_{xx},\quad A>0.
\end{gather}
(Note that here and in what follows we omit constant parameters that can be removed 
using shifts and scalings of independent and dependent variables and alternating their signs.) 
Such equations describe the propagation of weakly nonlinear longitudinal waves 
subject to thermoviscous diffusion and geometrical effects related to changing `ray tube area', $A(t)$,~\cite{HammertonCrighton1989}.
Each of these equations is mapped by the point transformation 
$\tilde t=\int\sqrt{A(t)}\,{\rm d}t$, $\tilde x=x$, $\tilde u=\sqrt{A(t)}u$ to the equation
\begin{gather}\label{eq:SimilarToNonplanarBurgersEq}
\tilde u_{\tilde t}+\tilde u\tilde u_{\tilde x}=g(\tilde t)\tilde u_{\tilde x\tilde x}
\end{gather}
with $g(\tilde t)=\sqrt{A(t)}$. 
The equations of the form~\eqref{eq:SimilarToNonplanarBurgersEq} are also called nonplanar Burgers equations 
\cite[Section~4.6]{Sachdev2009} and constitute 
the subclass of the class~$\mathcal L$ singled out by the constraints $C=1$ and $A^2_x=0$, 
which was, as mentioned above, intensively studied within the framework of group analysis of differential equations. 
The preimages of the equations~\eqref{eq:SimilarToNonplanarBurgersEq} with $g(\tilde t)=1,\tilde t/2,\exp\tilde t$ 
among equations of the form~\eqref{eq:NonplanarBurgersEq}
are respectively those with $A(t)\varpropto t^j$, $j=0,1,2$, 
which model plane ($j=0$), cylindrical ($j=1$) and spherical ($j=2$) 
unidirectional spreading of finite-amplitude waves, see~\cite{CrightonScott1979} and \cite[Sections~3.3--3.4]{Sachdev2009}. 
When the ratio~$g(\tilde t)/\tilde t$ respectively tends to infinity, a nonzero finite value or zero as $\tilde t\rightarrow\infty$, 
it is regarded as the super-cylindrical, cylindrical or subcylindrical case~\cite{Scott1981}.
If a perfect gas has an exponential mean density distribution $\rho(t)\varpropto\exp(-t/H)$, 
where the scale height~$H$ is large compared with the basic length scale of a wave propagating vertically upwards (the $t$-direction) 
and $x$ is the retarded time, then the velocity~$u(t,x)$ satisfies the equation~(47) in~\cite{CrightonScott1979}
whose simplified form is $u_t+uu_x-u/(2H)=\exp(t/H)u_{xx}$ and which is reduced by a point transformation 
to the equation of the form~\eqref{eq:SimilarToNonplanarBurgersEq} with $g(\tilde t)=\tilde t$.
For each nonvanishing function~$a$ of~$t$, the equation $u_t+uu_x-a_t(t)u/a(t)=a(t)u_{xx}$,
which models the acoustic waves in the atmosphere~\cite{Romanova1970},
is reduced by a point transformation to the standard Burgers equation.
The generalized Burgers equations $u_t+uu_x=\varepsilon xu_{xx}+g(x)u_x$ appear 
in the course of modeling ionised gases~\cite[Section~4.9]{Sachdev2009}.

The remainder of the paper is organized as follows.
In Section~\ref{Qu:sec:EquivGroup} we find the structure of the equivalence groupoid of the class~$\mathcal L$
by studying the imaged class~$\hat{\mathcal L}$.
The latter class can be represented as the disjoint union of two subclasses~$\hat{\mathcal L}_0$ and~$\hat{\mathcal L}_1$,
which are invariant under admissible transformations of the class~$\hat{\mathcal L}$.
These subclasses are normalized in the extended generalized and the usual senses, respectively.
We also show that~$\hat{\mathcal L}_0$ is weakly similar to a subclass of~$\mathcal L$ normalized in the usual sense.
The induced partition of the class~$\mathcal L$ into subclasses~$\mathcal L_0$ and~$\mathcal L_1$
is associated with a nontrivial differential constraint on the coefficients~$C$ and~$A^2$.
With all this knowledge at our disposal, in Section~\ref{Qu:sec:GroupClass} we efficiently carry out
the group classification of the class~$\hat{\mathcal L}_1$ with respect to its equivalence group using the algebraic method, and
that of~$\mathcal L_1$ with respect to its equivalence groupoid using the mapping technique.
The group classification of the subclass~$\mathcal L_0$ is reduced to the classification problem solved in~\cite{PocheketaPopovych2017}.
To obtain the group classification of the class~$\mathcal L$ with respect to its equivalence groupoid,
we merely concatenate the classification lists for the subclasses~$\mathcal L_0$ and~$\mathcal L_1$.
Section~\ref{Qu:sec:Conclusion} is left for conclusions, where we summarize the most essential findings of the paper. 
In the appendix, we supplement the theory of the algebraic method of group classification with 
assertions on the equivalence group and the equivalence algebra 
of a subclass of a general class of differential equations, 
whose systems admit the projection of the same subalgebra of the equivalence algebra of the entire class
as their common Lie invariance algebra. 
On the basis of these assertions, we suggest a procedure for gauging the arbitrary elements 
of such subclasses by equivalence transformations, 
which is needed to construct of a proper group classification list for the entire class.

\section{Equivalence groupoid}\label{Qu:sec:EquivGroup}

Let $\mathcal L_\theta$ denote a system of differential equations of the form $L(x,u^{(\EqOrd)},\theta^{(q)}(x,u^{(\EqOrd)}))=0$,
where $x$, $u$ and $u^{(r)}$ are the tuples of independent variables, of dependent variables
and of derivatives of~$u$ with respect to~$x$ up to order~$\EqOrd$.
The tuple of functions $L=(L^1,\dots,L^l)$ of~$(x,u^{(\EqOrd)},\theta)$  is fixed
whereas the tuple of functions~$\theta=(\theta^1,\dots,\theta^k)$ of~$(x,u^{(\EqOrd)})$
runs through the solution set~$\mathcal S$ of an auxiliary system
of differential equations and inequalities in~$\theta$, where $x$ and $u^{(\EqOrd)}$ jointly play the role of independent variables.
Thus, the \textit{class of (systems of) differential equations}~$\mathcal L$
is the parameterized family of systems~$\mathcal L_\theta$ with~$\theta$ running through the set~$\mathcal S$.
The components of~$\theta$ are called the arbitrary elements of the class~$\mathcal L$.
The equivalence groupoid of the class~$\mathcal L$ consists of admissible transformations of this class,
i.e., of triples of the form~$(\theta,\varphi,\tilde\theta)$.
Here~$\mathcal L_\theta$ and~$\mathcal L_{\tilde \theta}$
are the source and the target systems, belonging to the class~$\mathcal L$ and
corresponding to the values~$\theta$ and~$\tilde \theta$ of the arbitrary-element tuple, respectively,
which in turn runs through the solution set of the auxiliary system~$\mathcal S$,
and~$\varphi$ is a point transformation relating the equations~$\mathcal L_\theta$ and~$\mathcal L_{\tilde \theta}$.
The class~$\mathcal L$ is referred to as normalized in the usual (resp.\ generalized or extended generalized) sense
if its usual (resp.\ generalized or extended generalized) equivalence group
generates the entire equivalence groupoid of this class.
For more details, see~\cite{BihloCardosoPopovych2012,OpanasenkoBihloPopovych2017,PopovychKunzingerEshraghi2010}.

Below, by~$\mathcal L$ we denote the \textit{initial} class of variable-coefficient Burgers equations
of the general form~\eqref{Qu:QuClass}.
We begin the study of admissible transformations of $\mathcal L$
by considering its superclass~$\mathcal B$ of the more general Burgers equations
\begin{gather*}
u_t+C(t,x)uu_x=A^2(t,x)u_{xx}+A^1(t,x)u_x+A^0(t,x)u+B(t,x),
\end{gather*}
where the arbitrary elements $A$'s, $B$ and~$C$ are smooth functions of~$(t,x)$ with~$A^2C\neq0$.
The class~$\mathcal B$ is normalized in the usual sense~\cite{OpanasenkoBihloPopovych2017}.
Moreover, it is easy to present its equivalence group explicitly, 
unlike the equivalence group of its arbitrary-order counterpart constituted by general Burgers--Korteweg--de Vries equations,
where invoking the general Leibniz rule and Fa\`{a} di Bruno's formula 
makes the explicit expressions for equivalence transformations overly cumbersome.

\begin{proposition}
The class~$\mathcal B$ is normalized in the usual sense. 
Its usual equivalence group is constituted by the following point transformations in the relevant space:%
\footnote{%
In general, the equivalence group of a class of differential equations is defined to act
on the space coordinatized with the corresponding independent and dependent variables,
the derivatives of dependent variables up to the order of equations in the class
and class' arbitrary elements \cite{popo06a,PopovychKunzingerEshraghi2010}.
Since the arbitrary elements of all the classes of differential equations considered in this paper
are functions of the independent variables $(t,x)$ only,
we can assume that for each of these classes,
the corresponding equivalence group acts on the space coordinatized with the independent and dependent variables $(t,x,u)$
as well as the arbitrary elements of the class.
}
\begin{gather*}
\tilde t=T(t),\quad
\tilde x=X(t,x),\quad
\tilde u=U^1(t)u+U^0(t,x),
\\
\tilde A^0=\frac1{T_t}\left(A^0+\frac{U^0_x}{U^1}C+\frac{U^1_t}{U^1}\right),\quad
\tilde A^1=\frac{X_x}{T_t}\left(A^1+\frac{X_{xx}}{X_x}A^2+\frac{U^0}{U^1}C-\frac{X_t}{X_x}\right),\quad
\tilde A^2=\frac{X_x^2}{T_t}A^2,\\
\tilde B=\frac1{T_t}\left(B-\frac{U^0U^0_x}{U^1}C-U^0_{xx}A^2-U^0_xA^1-U^0A^0+U^0_t-\frac{U^0U^1_t}{U^1}\right),\quad
\tilde C=\frac{X_x}{T_tU^1}C,
\end{gather*}
where $T$, $X$, $U^0$ and~$U^1$ are smooth functions of their arguments with $T_tX_xU^1\neq0$.
\end{proposition}

To single out the class~$\mathcal L$ in the class~$\mathcal B$,
one needs to set additionally the arbitrary elements $A^0$, $A^1$ and~$B$ to be equal to zero,
i.e., to impose the constraints $A^0=0$, $A^1=0$ and~$B=0$.
At the same time, only two of these constraints --- either for $(A^0,A^1)$ or for $(A^1,B)$ --- can be simultaneously realized
via gauging the arbitrary elements of the class~$\mathcal B$ by a family of its equivalence transformations.
Moreover, none of the subclasses of~$\mathcal B$ obtained by successively setting the above constraints is normalized
in the usual sense, although the gauge~$A^1=0$ leads to a subclass normalized in the generalized sense,
cf.~\cite{OpanasenkoBihloPopovych2017}.

Hereafter for the sake of brevity we denote by~$(\tilde{*})$ an equation of the form~$(*)$ in the variables with tildes.
The knowledge of the equivalence groupoid of the superclass~$\mathcal B$ essentially simplifies
the computation of the equivalence groupoid~$\mathcal G^\sim$ of the class~$\mathcal L$.

\begin{proposition}\label{Qu:EquivGroupoidOfL}
A point transformation~$\varphi$ in the space with coordinates~$(t,x,u)$ connects equations~\eqref{Qu:QuClass} and~$(\tilde{\ref{Qu:QuClass}})$ in the class~$\mathcal L$
if and only if the components of~$\varphi$ are of the form
\begin{gather*}
\tilde t=T(t),\quad \tilde x=X(t,x),\quad \tilde u=U^1(t)u+U^0(t,x),
\end{gather*}
where~$T$, $X$, $U^0$ and~$U^1$ are smooth functions of their arguments with $T_tX_xU^1\neq0$,
satisfying the system of determining equations
\[
{X_t}=A^2X_{xx}+\frac{U^0 C}{U^1}X_x,\quad
CU^0_x=-U^1_t,\quad U^1U^0_t=A^2U^0_{xx}.
\]
The corresponding arbitrary-element tuples are related as follows%
\footnote{%
Throughout the paper, the left- and right-hand sides of such relations
are evaluated at the new variables, $(\tilde t,\tilde x)$, and at the old variables, $(t,x)$, respectively.
}
\begin{gather*}
\tilde C=\frac{X_x}{T_tU^1}C,\quad\tilde A^2=\frac{X_x^2}{T_t}A^2.
\end{gather*}
\end{proposition}

\begin{corollary}
The usual equivalence group~$G^\sim$ of the class~$\mathcal L$ is constituted
by the point transformations of the form
\begin{gather*}
\tilde t=T(t),\quad \tilde x=X^1x+X^0,\quad \tilde u=U^1u,\quad
\tilde C=\frac{X^1}{T_tU^1}C,\quad\tilde A^2=\frac{(X^1)^2}{T_t}A^2,
\end{gather*}
where $T$ is an arbitrary smooth function of~$t$ and $X^0$, $X^1$ and $U^1$ are arbitrary constants
with $T_tX^1U^1\neq0$.
\end{corollary}

The description of the equivalence groupoid~$\mathcal G^\sim$ of the class~$\mathcal L$ in Proposition~\ref{Qu:EquivGroupoidOfL}
allows us to test the equivalence of equations in this class with respect to point transformations.
In particular, the following assertion holds.

\begin{corollary}\label{Qu:ReductionOfEqsFromLToBurgersEq}
An equation of the form~\eqref{Qu:QuClass} reduces by a point transformation to the classical Burgers equation,
which is of the form~$(\tilde{\ref{Qu:QuClass}})$ with $\tilde A^2=\tilde C=1$,
if and only if the arbitrary elements~$A^2$ and~$C$ satisfy the constraints
\[
\left(\frac{(C)^2}{A^2}\right)_x=0,\quad
\left(\frac{C}{A^2}\right)_t=-C_{xx}.
\]
A point transformation realizing this equivalence can be chosen with parameter values defined by
$T_t=(C)^2/A^2$, $X_x=C/A^2$, $U^1=1$ and $U^0=0$.
\end{corollary}

To efficiently carry out the group classification of~$\mathcal L$,
we would like to take advantage of the algebraic method of group classification.
Nonetheless, the class~$\mathcal L$ is not normalized.
The common strategy to overcome such a difficulty in general
is to resort either to a partition of the class under consideration into normalized subclasses
\cite{PopovychKunzingerEshraghi2010,VaneevaPosta2017,WinternitzGazeau1992}
or to a mapping of this class to a class with better transformational properties \cite{VaneevaPopovychSophocleous2009,vaneeva2}.
In the present paper we successfully combine these two methods.

To begin with, we transform the equations of the class~$\mathcal L$ within the class~$\mathcal B$
by the family~$\mathcal F$ of point transformations $\hat t = t$, $\hat x=X(t,x):=\int (1/C(t,x))\mathrm dx$, $\hat u=u$
parameterized by the arbitrary element~$C$.
As a result, we obtain the \textit{imaged} class~$\hat{\mathcal L}$ consisting of equations of the form~\eqref{Qu:QuClassMapped},
\[
\hat{\mathcal L}\colon\quad\hat u_{\hat t}+\hat u\hat u_{\hat x}=\hat A^2(\hat t,\hat x)\hat u_{\hat x\hat x}+
\hat A^1(\hat t,\hat x)\hat u_{\hat x}.
\]
In particular, the arbitrary elements of source and target equations are related as
\[
\hat A^2=X_x^2A^2,\quad \hat A^1=X_{xx}A^2-X_t.
\]
Hereafter we omit hats over the variables and the arbitrary elements.
The class~$\hat {\mathcal L}$ can be singled out by the additional constraints
$A^0=0$ and~$B=0$ on the arbitrary elements of the subclass of~${\mathcal B}$ associated with the condition~$C=1$,
whose equivalence groupoid is described in~\cite[Theorem~4]{OpanasenkoBihloPopovych2017}.
This easily leads to the equivalence groupoid~$\hat{\mathcal G}^\sim$ of~$\hat{\mathcal L}$.

\begin{proposition}
A point transformation connects two equations in the class~$\hat{\mathcal L}$
if and only if its components are of the form
\begin{gather*}
\tilde t=T,\quad \tilde x=T_tU^1x+X^0,\quad \tilde u=U^1u-U^1_tx+U^0,
\end{gather*}
where~$T$, $X^0$, $U^0$ and~$U^1$ are smooth functions of~$t$ with $T_tU^1\neq0$
and
\[
U^1_{tt}x-U^0_t=A^1U^1_t.
\]
The arbitrary elements $(\tilde A^1,\tilde A^2)$ of the target equation
are expressed via the arbitrary elements $(A^1,A^2)$ of the source equation~as
\begin{gather*}
\tilde A^1=U^1A^1-U^1_tx+U^0-\frac{(T_tU^1)_tx+X^0_t}{T_t},\quad
\tilde A^2=T_t(U^1)^2A^2.
\end{gather*}
\end{proposition}

\begin{corollary}
The usual equivalence group~$ {\hat G}^\sim$ of the class~$\hat{\mathcal L}$
coincides with the generalized equivalence group thereof
and is constituted by the point transformations of the form
\begin{gather}\label{Qu:eq:hatLEquivTrans}
\begin{split}
&\tilde t=T,\quad \tilde x=T_tU^1x+X^0,\quad \tilde u=U^1u+U^0,\\
&\tilde A^1=U^1A^1+U^0-\frac{T_{tt}U^1x+X^0_t}{T_t},\quad
 \tilde A^2=T_t(U^1)^2A^2,
\end{split}
\end{gather}
where $T$ and~$X^0$ are arbitrary smooth functions of~$t$ with $T_t\neq0$,
and $U^0$ and~$U^1$ are arbitrary constants with~$U^1\neq0$.
\end{corollary}

The structure of the set of admissible transformations with a fixed source equation within the class~$\hat{\mathcal L}$
essentially depends on whether the arbitrary element~$A^1$ is affine with respect to~$x$.
The class~$\hat{\mathcal L}$ can be represented as the disjoint union of two subclasses,
the \textit{imaged singular} subclass~$\hat{\mathcal L}_0$
and the \textit{imaged regular} subclass~$\hat{\mathcal L}_1$,
which are singled out from the class~$\hat{\mathcal L}$ by the constraints~$A^1_{xx}=0$ and~$A^1_{xx}\neq0$, respectively.
Since these constraints are readily seen to be invariant under the admissible transformations of the class~$\hat{\mathcal L}$,
equations from the subclass~$\hat{\mathcal L}_0$ are not related to equations from the subclass~$\hat{\mathcal L}_1$
by point transformations.
As proved below, these subclasses are normalized in the extended generalized sense and in the usual sense, respectively.
Note that the constraints $A^1_{xx}=0$ and~$A^1_{xx}\neq0$ splitting the class~$\hat{\mathcal L}$ are much simpler than
their counterparts similarly splitting the class~$\mathcal L$
into the \textit{singular} and the \textit{regular} subclasses~$\mathcal L_0$ and~$\mathcal L_1$,
which are the preimages of the subclasses~$\hat{\mathcal L}_0$ and~$\hat{\mathcal L}_1$ with respect to
the family~$\mathcal F$ of point transformations, respectively,
\begin{gather}\label{Qu:SplittingCondition}
\mathcal L_0\colon\bigg(\frac{C_t}C-C\bigg({A^2}\frac{C_x}{C^2}\bigg)_{\!x\ }\bigg)_{\!x}=0\quad \text{and}\quad
\mathcal L_1\colon\bigg(\frac{C_t}C-C\bigg({A^2}\frac{C_x}{C^2}\bigg)_{\!x\ }\bigg)_{\!x}\neq0.
\end{gather}

The following corollary justifies the characteristic \textit{regular} for the subclass~$\hat{\mathcal L}_1$.

\begin{proposition}\label{Qu:prop1}
The class~$\hat{\mathcal L}_1$ is normalized in the usual sense. 
Its equivalence group~$\hat G^\sim_1$ coincides with~$\hat G^\sim$.
\end{proposition}

\begin{corollary}\label{Qu:theorem:EquivAlgebra}
The equivalence algebra~$\hat{\mathfrak g}^\sim_1$ of the class~$\hat{\mathcal L}_1$ is given by
$\hat {\mathfrak g}^\sim_1=\langle\hat D(\tau),\,\hat S^0,\,\hat S^1,\,\hat P(\chi) \rangle$,
where~$\tau$ and~$\chi$ run through the set of smooth functions of~$t$, and
\begin{gather*}
\hat D(\tau)=\tau\p_t+\tau_tx\p_x-\tau_{tt}x\p_{A^1}+\tau_tA^2\p_{A^2},\quad
\hat S^0=x\p_x+u\p_u+A^1\p_{A^1}+2A^2\p_{A^2},\\
\hat S^1=\p_u+\p_{A^1},\quad
\hat P(\chi)=\chi\p_x-\chi_t\p_{A^1}.
\end{gather*}
\end{corollary}

Since the classes~$\mathcal L_1$ and~$\hat {\mathcal L}_1$ are related
by the family~$\mathcal F$ of simple point transformations,
we can easily recover the equivalence groupoid~$\mathcal G^\sim_1$ of~$\mathcal L_1$
from the equivalence groupoid of~$\hat {\mathcal L}_1$,
which is particularly convenient, taking into account the complicated condition
singling out~$\mathcal L_1$ from~$\mathcal L$.

\begin{proposition}\label{Qu:corollary1}
A point transformation in the space with coordinates~$(t,x,u)$ connects equations
in the class~$\mathcal L_1$ if and only if its components are of the form
\begin{gather*}
\tilde t=T(t),\quad \tilde x=X(t,x),\quad \tilde u=U^1u+U^0,
\end{gather*}
where~$T$ and $X$ are smooth functions of their arguments, while $U^0$ and~$U^1$
are arbitrary constants, with $T_tX_xU^1\neq0$,
and the function~$X$ satisfies the Kolmogorov equation
\[
{X_t}=A^2{X_{xx}}+\frac{U^0 C}{U^1}X_x.
\]
The corresponding arbitrary elements are related by
\begin{gather*}
\tilde C=\frac{X_x}{T_tU^1}C,\quad\tilde A^2=\frac{X_x^2}{T_t}A^2.
\end{gather*}
\end{proposition}

\begin{corollary}
The usual equivalence group~$ G^\sim_1$ of the class~$\mathcal L_1$ coincides with that of the class~$\mathcal L$.
\end{corollary}

We will show that in contrast to the subclass~$\hat{\mathcal L}_1$, the subclass~$\hat{\mathcal L}_0$ of the class~$\hat{\mathcal L}$,
singled out by the constraint~$A^1_{xx}=0$, is not normalized in the usual sense
and, because of precluding its superclass~$\hat{\mathcal L}$ to be normalized, merits the attribute \textit{singular}.
We reparameterize the class~$\hat{\mathcal L}_0$, assuming the coefficients of the representation
$A^1(t,x)=A^{11}(t)x+A^{10}(t)$ for $A^1$ in view of the constraint~$A^1_{xx}=0$ as the new arbitrary elements instead of~$A^1$.
Thus, in what follows the arbitrary-element tuple for the class~$\hat{\mathcal L}_0$ is $\theta=(A^{10},A^{11},A^2)$.
The arbitrary elements~$A^{10}$ and~$A^{11}$ satisfy the auxiliary equations $A^{10}_x=A^{11}_x=0$.

\begin{proposition}\label{Qu:theorem:EqGrupoidTildeL_0}
The equivalence groupoid~$\hat{\mathcal G}^\sim_0$ of the class~$\hat{\mathcal L}_0$ consists
of the triples~$(\theta,\varphi,\tilde\theta)$, where~$\theta$ and~$\tilde\theta$ denote
the tuples of arbitrary elements of the source and the target equations in the class~$\hat{\mathcal L}_0$,
and~$\varphi$ is a point transformation whose components are of the form
\begin{subequations}
\begin{gather}\label{Qu:Aux1}
\tilde t=T,\quad \tilde x=T_tU^1x+X^0,\quad\tilde u=U^1u-U^1_tx+U^{00},
\end{gather}
where $T$, $X^0$, $U^1$ and~$U^{00}$ are smooth functions of~$t$, satisfying $T_tU^1\neq0$ and
\begin{gather}\label{Qu:Aux2}
U^1_{tt}=A^{11}U^1_t,\quad U^{00}_t=-A^{10}U^1_t.
\end{gather}
In turn, the arbitrary elements of the source and target equations are related as follows
\begin{gather}
\begin{split}\label{Qu:Aux3}
&\tilde{A}^2=(U^1)^2T_tA^2,\quad \tilde{A}^{11}=\frac1{T_t}\left( A^{11}-\frac{T_{tt}}{T_t}-\frac{2U^1_t}{U^1}\right),\\
&\tilde{A}^{10}=U^1A^{10}-\frac{X^0_t}{T_t}+U^{00}-\frac{X^0}{T_t}\left(A^{11}-\frac{T_{tt}}{T_t}-\frac{2U^1_t}{U^1}\right).
\end{split}
\end{gather}
\end{subequations}
\end{proposition}

To construct the usual equivalence group~$\hat G^\sim_0$ of the class~$\hat{\mathcal L}_0$, 
we split the classifying conditions~\eqref{Qu:Aux2} for admissible transformations 
with respect to the arbitrary elements~$A^{10}$ and~$A^{11}$ and find~$U^1$ and~$U^{00}$ to be constants.
This means that the group~$\hat G^\sim_0$ consists of the point transformations
in the space with coordinates $(t,x,u,A^{10},A^{11},A^2)$ whose components are of the form~\eqref{Qu:Aux1}, \eqref{Qu:Aux3},
where $T$ and~$X^0$ are smooth functions of~$t$ and $U^1$ and~$U^{00}$ are arbitrary constants with $T_tU^1\neq0$.
Therefore, the group~$\hat G^\sim_0$ also coincides with~$\hat G^\sim$  
but the class~$\hat{\mathcal L}_0$ is not normalized in the usual sense, 
and the similar assertion holds for the subclass~$\mathcal L_0$ of~$\mathcal L$, 
cf.\ Proposition~\ref{Qu:prop1} and Corollary~\ref{Qu:corollary1}. 
On the other hand, introducing the virtual nonlocal arbitrary elements $Y^0$, $Y^1$ and~$Y^2$ defined by the equations
\begin{gather}\label{Qu:eq:virtual}
Y^0_t=A^{11},\quad Y^1_t= {\rm e}^{Y^0},\quad Y^2_t=A^{10}{\rm e}^{Y^0},
\end{gather}
we construct a covering of the auxiliary system for the arbitrary elements of the class~$\hat{\mathcal L}_0$.
(This is an application of techniques from the theory of nonlocal symmetries of differential equations~\cite[Section~5]{Bocharov1999}
in the context of classes of differential equations.)
By~$\bar{\mathcal L}_0$ we denote the class obtained by reparameterizing the class~$\hat{\mathcal L}_0$
with the extended tuple of the arbitrary elements $\bar{\theta}=(A^{10},A^{11},A^2,Y^0,Y^1,Y^2)$.
The class~$\bar{\mathcal L}^0$ will be shown to be normalized in the generalized sense.

\begin{corollary}\label{Qu:theorem:EqGrupoidHatL_0}
The equivalence groupoid of the class~$\bar{\mathcal L}_0$ consists
of the triples~$(\bar{\theta},\varphi,\tilde{\bar\theta})$,
where the arbitrary-element tuples $\bar{\theta}$ and $\tilde{\bar\theta}$ of the source and the target equations
are related by~\eqref{Qu:Aux3}~and
\begin{gather}\label{Qu:Aux4}
\begin{split}
&\tilde Y^0=Y^0+\ln\frac{\delta}{T_t(c_1Y^1+c_0)^2}, \quad \tilde Y^1=\frac{c_1'Y^1+c_0'}{c_1Y^1+c_0},\\
&\tilde Y^2=\frac{\delta Y^2}{c_1Y^1+c_0}-\frac{\delta X^0{\rm e}^{Y^0}}{T_t(c_1Y^1+c_0)^2}-c_2\frac{c_1'Y^1+c_0'}{c_1Y^1+c_0}+c_3,
\end{split}
\end{gather}
and the components of the point transformation~$\varphi$ are the form~\eqref{Qu:Aux1}
with
\begin{gather}\label{Qu:Aux5}
U^1=c_1Y^1+c_0,\quad U^{00}=c_2-c_1Y^2,
\end{gather}
$\delta=c_1'c_0-c_1c_0'$,
$T$ and~$X^0$ being arbitrary smooth functions of~$t$ and $c$'s being arbitrary constants such that $\delta T_t>0$.
\end{corollary}

\begin{proof}
Having introduced the virtual arbitrary elements, we can solve the equations~\eqref{Qu:Aux2}
for~$U^1$ and~$U^{00}$ in terms of~$Y$'s.
The expression for the transformed nonlocal arbitrary element~$\tilde Y^0$ follows from the chain of identities
\[
\p_t\tilde Y^0=\tilde Y^0_{\tilde t}T_t=\tilde A^{11} T_t= A^{11}-\frac{T_{tt}}{T_t}-\frac{2c_1{Y^1_t}}{c_1Y^1+c_0}=
\left(Y^0+\ln\frac1{|T_t|(c_1Y^1+c_0)^2}\right)_t.
\]
For~$Y^1$ and~$Y^2$, the procedure is similar.
\end{proof}

\begin{remark}
There is a nontrivial gauge equivalence amongst equations in the reparameterized class~$\bar{\mathcal L}_0$
stemming from the indeterminacy in defining the virtual arbitrary elements. More specifically, the arbitrary-elements
tuples~$\bar\theta$ and $\tilde{\bar\theta}$ are associated with the same equation in the class~$\bar{\mathcal L}_0$
if and only if
\begin{gather}\label{Qu:Nomer}
\begin{split}
&\tilde A^{10}=A^{10},\quad \tilde A^{11}=A^{11},\quad \tilde A^2=A^2,\\ 
&\tilde Y^0=Y^0+\ln c_1',\quad \tilde Y^1=c_1'Y^1+c_0',\quad \tilde Y^2=c_1'Y^2+c_3,
\end{split}
\end{gather}
\looseness=1
where $c$'s are arbitrary constants with $c_1'>0$.
The equations~\eqref{Qu:Nomer} jointly with the equations $\tilde t=t$, $\tilde x=x$ and $\tilde u=u$
represent the components of the gauge equivalence transformations in~$\bar{\mathcal L}_0$,
which constitute the gauge equivalence group~\smash{$G^{\mathrm g\sim}_{\bar{\mathcal L}_0}$} of~$\bar{\mathcal L}_0$.
This group is a normal subgroup of the usual equivalence group~\smash{$G^{\sim}_{\bar{\mathcal L}_0}$} of~$\bar{\mathcal L}_0$,
and the quotient group~\smash{$G^{\sim}_{\bar{\mathcal L}_0}/G^{\mathrm g\sim}_{\bar{\mathcal L}_0}$} is isomorphic to
the usual equivalence group of the subclass~$\hat{\mathcal L}_0$ of~$\hat{\mathcal L}$, 
which coincides with the usual equivalence group of the entire class~$\hat{\mathcal L}$.
See details on gauge equivalence transformations in~\cite{PopovychKunzingerEshraghi2010}.
\end{remark}

\begin{theorem}
The class~$\bar{\mathcal L}_0$ is normalized in the generalized sense. 
Its generalized equivalence group~$\bar G^\sim_0$ consists of the point transformations of the form
\begin{subequations}\label{Qu:L0GenEquivTrans}
\begin{gather}
\label{Qu:L0GenEquivTrans1}
\tilde t=\bar T,\quad 
\tilde x=(\bar{\mathrm D}_t\bar T)(c_1Y^1+c_0)x+\bar X^0,\quad
\tilde u=(c_1Y^1+c_0)u-c_1{\rm e}^{Y^0}x+c_2-c_1Y^2,
\\ \label{Qu:L0GenEquivTrans2}
\tilde A^2=(\bar{\mathrm D}_t\bar T)(c_1Y^1+c_0)^2A^2,\quad
\tilde A^{11}=\frac1{\bar{\mathrm D}_t\bar T}
\left( A^{11}-\frac{\bar{\mathrm D}_t^2\bar T}{\bar{\mathrm D}_t\bar T}-\frac{2c_1 {\rm e}^{Y^0}}{c_1Y^1+c_0}\right),
\\ \label{Qu:L0GenEquivTrans3}
\tilde A^{10}=(c_1Y^1+c_0){A^{10}}-\frac{\bar{\mathrm D}_t\bar X^0}{\bar{\mathrm D}_t\bar T}+c_2-c_1Y^2-
\frac{\bar X^0}{\bar{\mathrm D}_t\bar T}\left(A^{11}-\frac{\bar{\mathrm D}_t^2\bar T}{\bar{\mathrm D}_t\bar T}-
\frac{2c_1 {\rm e}^{Y^0}}{c_1Y^1+c_0}\right),
\\ \label{Qu:L0GenEquivTrans4}
\tilde Y^0=Y^0+\ln\frac{\delta}{(\bar{\mathrm D}_t\bar T)(c_1Y^1+c_0)^2}, \quad \tilde Y^1=\frac{c_1'Y^1+c_0'}{c_1Y^1+c_0},
\\ \label{Qu:L0GenEquivTrans5}
\tilde Y^2=\frac{\delta Y^2}{c_1Y^1+c_0}-\frac{\delta \bar X^0{\rm e}^{Y^0}}{(\bar{\mathrm D}_t\bar T)(c_1Y^1+c_0)^2}-c_2\frac{c_1'Y^1+c_0'}{c_1Y^1+c_0}+c_3.
\end{gather}
\end{subequations}
Here  $\bar{\mathrm D}_t=\p_t+A^{11}_t\p_{A^{11}}+A^{10}_t\p_{A^{10}}+
A^{11}\p_{Y^0}+{\rm e}^{Y^0}\p_{Y^1}+A^{10}{\rm e}^{Y^0}\p_{Y^2}+A^2_t\p_{A^2}$
is the restricted total derivative operator with respect to~$t$,
$\delta:=c_1'c_0-c_0'c_1$, $\bar T$ and~$\bar X^0$ are smooth functions of~$(t,Y^1)$ and~$(t,Y^0,Y^1,Y^2)$, respectively,
and~$c$'s are arbitrary constants with $\delta\bar{\mathrm D}_t\bar T>0$.
\end{theorem}

\begin{proof}
Elements of the group~$\bar G^\sim_0$ are point transformations in the space with the coordinates $(t,x,u,A^{10},A^{11},A^2,Y^0,Y^1,Y^2)$.
Each of these transformations, $\mathcal T$, generates a family of admissible transformations of the class~$\bar{\mathcal L}_0$ 
with the following properties:
\begin{itemize}\itemsep=-0.5ex
\item[$\circ$]
they are smoothly and pointwise parameterized by the source arbitrary-element tuple~$\bar{\theta}$,
\item[$\circ$]
their transformational parts are of the general form~\eqref{Qu:Aux1},
\item[$\circ$]
their target arbitrary-element tuples are related to the source ones according to~\eqref{Qu:Aux3} and~\eqref{Qu:Aux4},
\item[$\circ$]
and the parameters~$U^1$ and~$U^{00}$ in them are necessarily of the form~\eqref{Qu:Aux5}.
\end{itemize}
Therefore, the components of~$\mathcal T$ are of the form~\eqref{Qu:L0GenEquivTrans},
where the parameters $\bar T$, $\bar X^0$ and $c$'s are considered as smooth functions of the above coordinates
that satisfy the equations
\begin{gather*}
\bar{\mathrm D}_x\bar T=\bar{\mathrm D}_u\bar T=0,\quad
\bar{\mathrm D}_x\bar X^0=\bar{\mathrm D}_u\bar X^0=0,\\
\bar{\mathrm D}_tc_i=\bar{\mathrm D}_xc_i=\bar{\mathrm D}_uc_i=0,\ i=0,\dots,3,\quad
\bar{\mathrm D}_tc_j'=\bar{\mathrm D}_xc_j'=\bar{\mathrm D}_uc_j'=0,\ j=0,1,
\end{gather*}
with $\bar{\mathrm D}_t$ defined in the theorem's statement,
$\bar{\mathrm D}_x:=\p_x+A^2_x\p_{A^2}$ and $\bar{\mathrm D}_u:=\p_u$.
Successively splitting these equations with respect to~$A^2_x$,
and then the equations for~$c$'s with respect to~$A^{10}_t$, $A^{11}_t$, $A^{10}$, $A^{11}$ and $Y^0$
(the last three splittings are allowed in view of equations derived in the course of the previous splittings),
we get that
$\bar T$ and $\bar X^0$ are smooth functions of~$(t,A^{10},A^{11},Y^0,Y^1,Y^2)$,
and~$c$'s are constants.
After this, we also split the equations~\eqref{Qu:L0GenEquivTrans2} and~\eqref{Qu:L0GenEquivTrans3}
with respect to~$A^{10}_t$ and~$A^{11}_t$,
obtaining
$\bar T_{A^{10}}=\bar T_{A^{11}}=\bar T_{Y^0}=\bar T_{Y^2}=0$ and $\bar X^0_{A^{10}}=\bar X^0_{A^{11}}=0$,
which completes the proof.
\end{proof}

Given a class~$\mathcal C$ of differential equations, consider subgroups of its generalized equivalence group
that generate the same equivalence groupoid of~$\mathcal C$ as the entire group does.
A minimal subgroup of the above kind is called an effective generalized equivalence group of~$\mathcal C$~\cite{OpanasenkoBihloPopovych2017}.
An effective generalized equivalence group may coincide with the entire generalized equivalence group.
For instance, this is the case when the generalized equivalence group coincides with the usual one.
An effective generalized equivalence group is nontrivial if it is a proper subgroup of the corresponding generalized equivalence group.

\looseness=-1
Note that should we merely omit the dependence of the group parameters~$\bar T$ and $\bar X^0$ in~$\bar G^\sim_0$ on the nonlocal arbitrary elements~$Y$'s,
we would obtain the set of equivalence transformations that is not a group as it is not closed under the composition of transformations
although this set still generates the entire equivalence groupoid of~$\bar{\mathcal L}_0$.
In particular, the value~$\bar X^{0,3}$ of the parameter function~$\bar X^0$ for the composition~$\mathcal T^3$ of transformations~$\mathcal T^1,\mathcal T^2\in \bar G^\sim_0$
would be of the form
\begin{gather}\label{Qu:eq:composition}
\bar X^{0,3}=\mathrm D_{\tilde t} \bar T^{,2}(c_{1,2} \tilde Y^1+ c_{0,2})\bar X^{0,1}+\bar X^{0,2},
\end{gather}
where an index after comma indicates the number of the transformation the parameters are associated with.
Thus, the dependence of~$\bar X^0$ on~$Y^1$ is necessary for closedness with respect to the composition of the transformations.
In a similar way, we can show that the parameter~$\bar X^0$ should depend on~$Y^0$.
At the same time, the dependence of~$\bar T$ on the virtual arbitrary elements
as well as the dependence of~$\bar X^0$ on~$Y^2$ are superfluous.
Guided by inspection and intuition, we look for transformations
with the parameter~$\bar X^0$ of the form $\bar X^0=T_t \exp(\alpha Y^0) (c_1Y^1+c_0)^\beta\breve X^0(t)$ for some constants~$\alpha$ and~$\beta$.
The substitution of the ansatz into~\eqref{Qu:eq:composition} readily produces $\alpha=-1/2$ and~$\beta=1$.

\begin{corollary}
An effective generalized equivalence group~$\breve G^\sim_0$ of the class~$\bar{\mathcal L}_0$ is constituted by the point transformations
\begin{gather*}
\tilde t=T,\quad \tilde x=T_t(c_1Y^1+c_0)\left(x+{\rm e}^{-{Y^0}/2}\breve X^0\right),\quad \tilde u=(c_1Y^1+c_0)u-c_1{\rm e}^{Y^0}x+c_2-c_1Y^2,\\
\tilde{A}^{10}=(c_1Y^1+c_0)\left(A^{10}-\frac12{\rm e}^{-{Y^0}/{2}}\breve X^0A^{11}-{\rm e}^{-{Y^0}/{2}}\breve X^0_t\right)+
{c_1{\rm e}^{{Y^0}/2}}+c_2-c_1Y^1,\\
\tilde{A}^{11}=\frac1{T_t}\bigg( A^{11}-\frac{T_{tt}}{T_t}-\frac{2c_1 {\rm e}^{Y^0}}{c_1Y^1+c_0}\bigg),\quad
\tilde{A}^2={T_t{(c_1Y^1+c_0)^2}A^2},\\
\tilde Y^0=Y^0+\ln\frac{\delta}{T_t(c_1Y^1+c_0)^2}, \quad \tilde Y^1=\frac{c_1'Y^1+c_0'}{c_1Y^1+c_0},\\
\tilde Y^2=\frac{\delta Y^2}{c_1Y^1+c_0}-\frac{\delta \breve X^0{\rm e}^{{Y^0}/2}}{c_1Y^1+c_0}-c_2\frac{c_1'Y^1+c_0'}{c_1Y^1+c_0}+c_3,
\end{gather*}
where $\delta:=c_1'c_0-c_0'c_1$, $T$ and~$\breve X^0$ are smooth functions of~$t$ and~$c$'s are arbitrary constants with $\delta T_t>0$.
\end{corollary}

\begin{proof}
To prove that the set of transformations from the corollary's statement is 
an effective generalized equivalence group of the class~$\bar{\mathcal L}_0$,
one should show that 
it is indeed a group under the composition of transformations, 
it induces the entire equivalence groupoid of the class~$\bar{\mathcal L}_0$, and 
it is a minimal group with this property. 
The first statement is proved by mere inspection,
while the second (two-part) statement is more involved.
Given an equation~$\bar{\mathcal L}_0^{\bar\theta}$ in the class~$\bar{\mathcal L}_0$
with a fixed value of the tuple of arbitrary elements~$\bar\theta$,
the set~$\mathrm T_{\bar\theta}$ of admissible transformations with the source~$\bar\theta$
is parameterized by arbitrary smooth functions~$T$ and~$X^0$ of~$t$
and arbitrary constants~$c_0$, \dots, $c_3$, $c_0'$ and~$c_1'$ such that~$(c_1'c_0-c_1c_0')T_t>0$.
At the same time, each admissible transformation in~$\mathrm T_{\bar\theta}$ is generated
by the element from~$\breve G^\sim_0$ with the same values of all the parameters except $\breve X^0$
whose value is defined by $\breve X^0=X^0{\rm e}^{Y^0/2}/(T_t(c_1Y^1+c_0))$ 
with the fixed values of the arbitrary elements~$Y^0$ and~$Y^1$.
This establishes a one-to-one correspondence between the group~$\bar G^\sim_0$ and~$\mathrm T_{\bar\theta}$, 
completing the proof.
\end{proof}

\begin{corollary}
The class~$\hat{\mathcal L}_0$ is normalized in the extended generalized sense. 
Its equivalence groupoid is generated by the group~$\breve G^\sim_0$.
\end{corollary}

Moreover, the class~$\hat{\mathcal L}_0$ can be mapped by the family of equivalence transformations
with $U^1=1$, $U^{00}=0$ and the para\-meters~$T$ and~$X^0$ satisfying the system
\[
T_{tt}=A^{11}T_t,\quad X^0_t=A^{10}T_t
\]
to its subclass~$\mathcal L_{0'}$ of equations of the form
$u_t+uu_x=A^2(t,x)u_{xx}$ studied in~\cite{PocheketaPopovych2017}.
In other words, we gauge the arbitrary elements~$A^{10}$ and~$A^{11}$ to zero by admissible transformations.
The class~$\mathcal L_{0'}$ is the subclass of the initial class~$\mathcal L$ as well.
It is singled out from~$\mathcal L$ by the condition $C=1$, $\mathcal L_{0'}=\mathcal L\cap \hat{\mathcal L}$.
In other words, every equation in the class~$\mathcal L_0$
is mapped by a point transformation to an equation in the same class with $C=1$.
For completeness we present here the assertion from~\cite{PocheketaPopovych2017}
on the equivalence groupoid of the class~$\mathcal L_{0'}$.

\begin{proposition}\label{Qu:prop2}
The class~$\mathcal L_{0'}$ is normalized in the usual sense.
Its equivalence group~$G^\sim_{0'}$ is constituted by the point transformations of the form
\begin{gather*}
\tilde t=\frac{\alpha t+\beta}{\gamma t+\delta},\quad \tilde x=\frac{\kappa x+\mu_1t+\mu_0}{\gamma t+\delta},\quad
\tilde u=\frac{\kappa(\gamma t+\delta)u-\kappa\gamma x+\mu_1\delta-\mu_0\gamma}{\alpha\delta-\beta\gamma},\quad
\tilde A^2=\frac{\kappa^2 A^2}{\alpha\delta-\beta\gamma},
\end{gather*}
where $\alpha$, $\beta$, $\gamma$, $\delta$, $\kappa$, $\mu_0$ and~$\mu_1$ are arbitrary constants 
with $(\alpha\delta-\beta\gamma)\kappa\ne0$ that are defined up to a nonzero constant multiplier.
\end{proposition}

\section{Group classification}\label{Qu:sec:GroupClass}

Classes of differential equations are called similar if they are related by a point transformation of dependent and independent variables.
One of the distinguished features of similar classes is the fact that their group classifications with
respect to the corresponding equivalence groups (resp.\ groupoids)
are interrelated and can be obtained from one another via the so-called mapping method,
see more details in~\cite{VaneevaPopovychSophocleous2009}.
At the same time, the classes~$\mathcal L$ and~$\hat{\mathcal L}$ (and the corresponding subclasses)
are related only by a family~$\mathcal F$ of point transformations
parameterized by the arbitrary elements of source equations.
We call such classes \textit{weakly similar}.
One may suppose that the group classifications of the weakly similar classes are still directly related
but in fact this is not the case for group classifications with respect to equivalence groups.
Nevertheless, a similar statement about group classifications with respect to equivalence groupoids is possible.

Note that in view of results of the previous section,
the group classification of the class~$\mathcal L_0$ up to $\mathcal G^\sim_0$-equivalence reduces to
that of~$\mathcal L_{0'}$ up to  $G^\sim_{0'}$-equivalence,
and the latter classification was carried out in~\cite{PocheketaPopovych2017}.

Summing up the above consideration, we formulate the following theorem.

\begin{theorem}
A complete list of $\mathcal G^\sim$-inequivalent Lie-symmetry extensions within the class~$\mathcal L$ of variable-coefficient Burgers equations
is the disjoint union of a complete list of $G^\sim_{0'}$-inequivalent Lie-symmetry extensions of the normalized subclass~$\mathcal L_{0'}$
and a complete list of $\mathcal G^\sim_1$-inequivalent Lie-symmetry extensions of the subclass~$\mathcal L_1$.
The latter list is necessarily the preimage of a complete list
of $G^\sim_1$-inequivalent Lie-symmetry extensions of the normalized class~$\hat{\mathcal L}_1$
by the family~$\mathcal F$ of point transformations $\tilde t = t$, $\tilde x=\int (1/C(t,x))\mathrm dx$, $\tilde u=u$,
parameterized by the arbitrary element~$C$ of source equations,
which maps the class~$\mathcal L$ onto the class~$\hat{\mathcal L}$.
\end{theorem}

For the aforementioned reasons, in what follows we focus on the regular subclass~$\mathcal L_1$.
For this purpose, at first we need to consider the imaged counterpart~$\hat{\mathcal L}_1$ of~$\mathcal L_1$.

\subsection{Preliminary group analysis of the imaged regular subclass}

We carry out the group classification of the class~$\hat{\mathcal L}_1$
using the algebraic method within the framework of the infinitesimal approach~\cite{Ovsiannikov1982}.
Let~$\hat{\mathcal L}_1^\kappa$ be an equation in the class~$\hat{\mathcal L}_1$
with a fixed value of the arbitrary-element tuple~$\kappa=(A^1,A^{2})$.
The infinitesimal invariance criterion implies that
a vector field $Q=\tau \p_t+\xi \p_x+\eta \p_u$ with components being smooth functions of~$(t,x,u)$ is
the infinitesimal generator of a one-parameter point symmetry group of this equation if
\[
Q^{(2)}\left(u_t+uu_x-A^2u_{xx}-A^1u_x\right)\big|_{\hat{\mathcal L}_1^\kappa}=0,
\]
where~$Q^{(2)}=Q+\sum_{0<|\alpha|\leqslant2}\eta^{\alpha}\p_{u_{\alpha}}$ is the second prolongation of~$Q$,
$\alpha=(\alpha_1,\alpha_2)\in\mathbb N_0^2$ is a multiindex,
$|\alpha|=\alpha_1+\alpha_2$, $u_{\alpha}=\p^{|\alpha|}u/\p t^{\alpha_1}\p x^{\alpha_2}$,
and the coefficients~$\eta^\alpha$ are defined by the prolongation formula~\cite{Olver1993},
\[
\eta^\alpha=\mathrm D_t^{\alpha_1}\mathrm D_x^{\alpha_2}\left(\eta-\tau u_t-\xi u_x\right)+\tau u_{\alpha+(1,0)}+\xi u_{\alpha+(0,1)}
\]
with $\mathrm D_t$ and~$\mathrm D_x$ being the total derivative operators with respect to~$t$ and~$x$, respectively.
Expanding the above condition produces
\begin{gather}\label{Qu:InfinCrit}
\eta^{(1,0)}+\eta u_x+u\eta^{(0,1)}=(\tau A^2_t+\xi A^2_x)u_{xx} +(\tau A^1_t+\xi A^1_x)u_x + A^2\eta^{(0,2)}+ A^1\eta^{(0,1)},
\end{gather}
whenever $u_t=A^2u_{xx}+A^1u_x-uu_x$ holds. Since the class~$\hat{\mathcal L}_1$ is normalized,
we can impose additional restrictions on the vector field~$Q$ in view of Corollary~\ref{Qu:theorem:EquivAlgebra},
\[
\tau=\tau(t),\quad\xi=(\tau_t+\alpha)x+\chi,\quad\eta=\alpha u+\beta
\]
with constants~$\alpha$ and~$\beta$ and with smooth functions~$\tau$ and~$\chi$ of~$t$,
cf.~\cite{BihloCardosoPopovych2012,KurujyibwamiBasarabHorwathPopovych2018}.
On plugging the expressions for~$\tau$, $\xi$, $\eta$ and~$u_t$ in~\eqref{Qu:InfinCrit},
we split successively the resulting equation with respect to~$u_{xx}$ and~$u_x$ to get the system of classifying equations
\begin{subequations}\label{Qu:ClassEqs}
\begin{gather}
\tau A^2_t+((\tau_t+\alpha)x+\chi)A^2_x=(\tau_t+2\alpha)A^2,\label{Qu:ClassEqsA}\\
\tau A^1_t+((\tau_t+\alpha)x+\chi)A^1_x=\alpha A^1-\tau_{tt}x-\chi_t+\beta\label{Qu:ClassEqsB}.
\end{gather}
\end{subequations}
Splitting the system~\eqref{Qu:ClassEqs} with respect to the arbitrary elements~$A^1$ and~$A^2$
as well as their derivatives, we can show
that the kernel Lie invariance algebra of equations of the class~$\hat{\mathcal L}_1$
is the trivial zero algebra.
The preliminary description of the Lie invariance algebras of equations in the class~$\hat{\mathcal L}_1$ is as follows.

\begin{proposition}\label{Qu:theorem:PropAppropAlg}
The maximal Lie invariance algebra~$\mathfrak g_\kappa$ of the equation~$\hat{\mathcal L}_1^\kappa$
is spanned by vector fields of the form $Q=D(\tau)+\alpha S^1+\beta S^0+ P(\chi)$,
where the constants~$\alpha$ and~$\beta$ and the smooth functions~$\tau$ and~$\chi$ of~$t$
satisfy the classifying equations~\eqref{Qu:ClassEqs}, and
\[
D(\tau)=\tau\p_t+\tau_tx\p_x,\quad
S^1=x\p_x+u\p_u,\quad
S^0=\p_u,\quad
P(\chi)=\chi\p_x.
\]
\end{proposition}

To complete the solution of the group classification problem for the class~$\hat{\mathcal L}_1$,
we have to solve the classifying equations~\eqref{Qu:ClassEqs} up to $\hat {G}^\sim_1$-equivalence,
which is done using an advanced version of the algebraic method.
This method was invented by Sophus Lie in the course of classifying Lie symmetries of second-order ODEs.
Nevertheless, it was not fully understood back then
and forgotten until the 1990s, when it implicitly reemerged in~\cite{GagnonWinternitz1993,GazeauWinternitz1992}.
Later the practitioners appreciated the method and started to use it intensively in different disguises,
cf. \cite{BH,gung04a,huan09a,lahno,meleshko,Nikitin,popo2004a,zhda99Ay},
in particular, to carry out the preliminary group classification~\cite{Akhatov,card11Ay,ibra2}.
The theoretical background of the method was developed in~\cite{popo06a},
and it was realized that the method works best for normalized classes of differential equations
\cite{BihloCardosoPopovych2012,KurujyibwamiBasarabHorwathPopovych2018,OpanasenkoBihloPopovych2017,PopovychKunzingerEshraghi2010}.
Nowadays it has been applied to a wide variety of classes of differential equations.

\subsection{Properties of appropriate subalgebras of the imaged regular subclass}

Consider the span~$\hat{\mathfrak g}$ of the vector fields~$D(\tau)$, $P(\chi)$, $S^0$ and~$S^1$,
where $\tau$ and~$\chi$ run through the set of smooth functions of~$t$,
$\hat{\mathfrak g}=\langle D(\tau), P(\chi), S^0, S^1\rangle$.
It is endowed with the structure of Lie algebra with the Lie bracket of vector fields as the algebra multiplication.
The independent nonzero commutation relations between elements of~$\mathfrak g$ are exhausted by
\begin{gather*}
[D(\tau^1),D(\tau^2)]=D(\tau^1\tau^2_t-\tau^1_t\tau^2),\quad
[D(\tau),P(\chi)]=P(\tau\chi_t-\tau_t\chi),\\
[S^0,S^1]=S^0,\quad[P(\chi),S^1]=P(\chi).
\end{gather*}

Let~$\pi$ be the projection map from the joint space of the variables and the arbitrary elements of the class~$\hat{\mathcal L}_1$
to the space of the variables only, $\pi(t,x,u,A^1,A^2)=(t,x,u)$.
This map induces the well-defined pushforward~$\pi_*$ for related point transformations and vector fields
such that~$\pi_* \hat{\mathfrak g}^\sim_1=\hat{\mathfrak g}$
is the Lie algebra of the Lie (pseudo)group~$\pi_*\hat G^\sim_1$.
Recall that~$\hat{\mathfrak g}^\sim_1$ is the equivalence algebra of the class~$\hat{\mathcal L}_1$.
We call a subalgebra~$\mathfrak a\subseteq\hat{\mathfrak g}$ \emph{appropriate} for the class~$\hat{\mathcal L}_1$
if there exists a value of the arbitrary-element tuple~$\kappa$ such that $\mathfrak a=\hat{\mathfrak g}_\kappa$,
cf.~\cite{BihloCardosoPopovych2012,BihloPopovych2016,OpanasenkoBihloPopovych2017}.
It is natural to call a subalgebra~$\mathfrak s$ of~$\hat{\mathfrak g}^\sim_1$ \emph{appropriate} for the class~$\hat{\mathcal L}_1$ if
$\pi_*\mathfrak s$ is an appropriate subalgebra of~$\hat{\mathfrak g}$.
The pushforward~$\pi_*$ establishes a one-to-one correspondence between the appropriate subalgebras
of~$\hat{\mathfrak g}^\sim_1$ and of~$\hat{\mathfrak g}$.
Furthermore, there are well-defined adjoint actions of the Lie group~$\pi_*\hat G^\sim_1$
on the Lie algebra~$\hat{\mathfrak g}$ and, therefore, on the set of its subalgebras,
which are consistent with the actions of~$\hat G^\sim_1$ on the algebra~$\hat{\mathfrak g}^\sim_1$
and on the set of its subalgebras, respectively.
In particular, the above actions preserve the respective sets of appropriate subalgebras
of~$\hat{\mathfrak g}$ and of~$\hat{\mathfrak g}^\sim_1$.
As a result, the classifications of appropriate subalgebras
of~$\hat{\mathfrak g}$ and of~$\hat{\mathfrak g}^\sim_1$ are equivalent,
and we can consider any of them.

We refer to the transformations~$\mathcal D(T)$, $\mathcal P(X^0)$,
$\mathcal S^1(U^1)$ and~$\mathcal S^0(U^0)$ in~$\pi_*\hat G^\sim_1$,
obtained by restricting all but one of the parameter
functions $T$, $X^0$ and constants~$U^0$ and~$U^1$ in~\eqref{Qu:eq:hatLEquivTrans}
to the values corresponding to the identity transformation,
as the elementary equivalence transformations of the class~$\hat{\mathcal L}_1$.
The complete list of nontrivial adjoint actions of the elementary transformations
on the elements spanning~$\hat{\mathfrak g}$ is
\begin{gather*}
\mathcal D_*(T) D(\tau)=\tilde D(T_t\tau),\quad
\mathcal D_*(T) P(\chi)=\tilde P(T_t\chi),\\
\mathcal S^1_*(U^1) S^0=U^1 \tilde S^0,\quad
\mathcal S^1_*(U^1) P(\chi)=\tilde P(U^1 \chi),\quad
\mathcal S^0_*(U^0) S^1=\tilde S^1-U^0\tilde S^0,\\
\mathcal P_*(X^0) D(\tau)=\tilde D(\tau)-\tilde P(X^0\tau_t-X^0_t\tau),\quad
\mathcal P_*(X^0) S^1=\tilde S^1-\tilde P(X^0).
\end{gather*}
Here the tildes over the right-hand side vector field indicate that these vector fields are expressed
using the variables with tildes, which also includes substituting $t=\tilde T(\tilde t)$ for $t$,
where $\tilde T$ is the inverse of the function~$T$.

Thus, the problem of group classification of the class~$\hat {\mathcal L}_1$ has boiled down to the classification
of appropriate subalgebras of~$\hat{\mathfrak g}$ under the adjoint action of~$\pi_*\hat G^\sim_1$ thereon.
For an efficient classification we first need to show some properties of appropriate subalgebras,
which is done in a way similar to~\cite{BihloPopovych2016,OpanasenkoBihloPopovych2017}.
Let $k_1=k_1(\kappa)=\dim (\hat{\mathfrak g}_\kappa\cap\langle S^1, S^0, P(\chi) \rangle)$ and~$\varpi\colon(t,x,u)\to t$
be the projection map singling out the first component.
Note that $k_1$ is the $\pi_*\hat G^\sim_1$-invariant integer and
the pushforward~$\varpi_*$ is well defined on~$\hat{\mathfrak g}$,
with $\varpi_*\hat{\mathfrak g}=\{\tau\p_t\}$,
where~$\tau$ runs through the set of smooth functions of~$t$.

Properties of appropriate subalgebras are described in the following proposition.

\begin{proposition}\label{Qu:theorem:PropertiesAppropriateSub}
For any equation~$\hat{\mathcal L}_1^\kappa$ in the class~$\hat{\mathcal L}_1$, we have

1. $k_1=0\bmod\pi_*\hat G^\sim_1$;

2.\ $\varpi_* \hat{\mathfrak g}_\kappa$ is a Lie algebra with $k_2:=\dim \varpi_* \hat{\mathfrak g}_\kappa\leqslant 2$,
and
\[
\varpi_* \hat{\mathfrak g}_\kappa\in\big\{\{0\},\langle \p_t\rangle,\langle\p_t,t\p_t\rangle\big\}\bmod\varpi_*\pi_*\hat G^\sim_1.
\]
\end{proposition}
\begin{proof}

1. Let $\hat{\mathcal L}_1^\kappa$ admit a Lie symmetry~$\alpha S^1+\beta S^0+P(\chi)$,
where at least one of the parameters~$\alpha$, $\beta$ or $\chi$ does not vanish.
Then the equation~\eqref{Qu:ClassEqsB} yields that $\alpha$ or $\chi$ does not vanish.
Differentiating~\eqref{Qu:ClassEqsB} once with respect to~$x$, we derive that~$A^1_{xx}=0$,
which contradicts the auxiliary inequality~$A^1_{xx}\neq0$.

2. The assertion follows from Lie's classification of realizations of Lie algebras on the real line
in view of the facts that
$\varpi_*\hat{\mathfrak g}_\kappa$~is a finite-dimensional Lie algebra
for any arbitrary-element tuple~$\kappa$, cf.~\cite[Lemmas~15 and~18]{OpanasenkoBihloPopovych2017}.
Indeed, $\varpi_*\hat{\mathfrak g}_\kappa$ is a finite-dimensional subalgebra of $\varpi_*\hat{\mathfrak g}=\{\tau\p_t\}$,
and since $\varpi_*\pi_*\hat G^\sim_1$ coincides with the (pseudo)group of local diffeomorphisms on the $t$-line,
$\varpi_*\hat{\mathfrak g}_\kappa$ is $\varpi_*\pi_*\hat G^\sim_1$-equivalent to a subalgebra in $\{0,\langle \p_t\rangle,\langle\p_t,t\p_t\rangle,\langle\p_t,t\p_t,t^2\p_t\rangle\}$.

We show by contradiction that $\dim\varpi_*\hat{\mathfrak g}_\kappa<3$.
Suppose that an equation~$\hat{\mathcal L}^\kappa_1$ admits a three-dimensional invariance algebra isomorphic to~$\mathfrak{sl}_2(\mathbb R)$.
Then so does an equation~$\mathcal E$ in~$\mathcal B$ with $(C,A^1)=(1,0)$
that is $G^\sim_{\mathcal B}$-equivalent to~$\hat{\mathcal L}^\kappa_1$.
According to~\cite[Theorem~19]{OpanasenkoBihloPopovych2017},
the arbitrary-element tuple $(A^2,A^0,B)$ associated with~$\mathcal E$ is equal, up to~$G^\sim_{\mathcal B}$-equivalence, to
either~$(a_2,0,bx^{-3})$ or $(1,0,0)$.
Here~$a_2$ and~$b$ are arbitrary constants with $a_2\neq0$.
But in both the cases the equation~$\mathcal E$ is $G^\sim_{\mathcal B}$-equivalent to no equation in the class~$\hat{\mathcal L}_1$.
This easily follows from~\cite[Theorem~4]{OpanasenkoBihloPopovych2017}
since if an equation in~$\mathcal B$ with $A^1=A^0=0$ and $C=1$ is equivalent to an equation in~$\hat{\mathcal L}$,
then for the latter equation one necessarily has the constraint $A^1_{xx}=0$.
\end{proof}

\begin{corollary}
$\dim\hat{\mathfrak g}_\kappa\leqslant2$ for any $\kappa$.
\end{corollary}

\subsection{Group classification of the imaged regular subclass}

\begin{theorem}\label{thm:GroupClassificationGenBurgersKdVEqs}
A complete list of $\hat G^\sim_1$-inequivalent
Lie-symmetry extensions in the class~$\hat{\mathcal L}_1$ is exhausted by the cases given in Table~\ref{Qu:tab:CompleteGroupClassificationBurgersEquations}.
\end{theorem}

\begin{proof}
The improper appropriate subalgebra~$\mathfrak a=\{0\}$ corresponds to the general case of equations in the class~$\hat{\mathcal L}_1$
without Lie symmetries.

According to Proposition~\ref{Qu:theorem:PropertiesAppropriateSub}, a basis of each nontrivial appropriate subalgebra~$\mathfrak a$
of the algeb\-ra~$\hat{\mathfrak g}$ consists of~$k_2$ vector fields 
\[Q^i=D(\tau^i)+\alpha_i S^1+\beta_i S^0+P(\chi^i),\quad i=1,\dots,k_2,\]
with linearly independent~$\tau$'s, and $k_2\in\{1,2\}$.
We carry out the proof by separately considering both the cases for values of~$k_2$.
Taking into account the closedness under the Lie bracket
and the equivalence relation on the set of appropriate subalgebras,
we find a canonical representative for each equivalence class of appropriate subalgebras.
The successive solution of the classifying equations produces expressions for arbitrary elements of equations admitting
nontrivial Lie invariance algebras. The arbitrary elements are to be simplified using equivalence transformations
stabilizing the corresponding appropriate subalgebra in the sense below.

\begin{table}[!t]
\caption{Complete group classification of the class~$\hat{\mathcal L}_1$ up to $\hat G^\sim_1$-equivalence.
\label{Qu:tab:CompleteGroupClassificationBurgersEquations}}
\begin{center}\newcounter{tbn}\setcounter{tbn}{-1}
\def\arraystretch{1.5} 
\begin{tabular}{|c|l|l|l|}
\hline
no. &\hfil $A^2$ &\hfil $A^1$& \hfil Basis of~$\hat{\mathfrak g}_\kappa$\\
\hline
\refstepcounter{tbn}\thetbn\label{Qu:1}  & $A^2(t,x)$                       & $A^1(t,x)$                         & ---\\
\refstepcounter{tbn}\thetbn\label{Qu:2}  & $\phi(x)$                        & $\psi(x)$                          & $D(1)$\\
\refstepcounter{tbn}\thetbn\label{Qu:2b} & $\phi(x)$                        & $\psi(x)+t$                        & $D(1)+S^0$\\
\refstepcounter{tbn}\thetbn\label{Qu:3}  & ${\rm e}^{-2t}\phi(x{\rm e}^t)$  & ${\rm e}^{-t}\psi(x{\rm e}^t)$     & $D(1)-S^1$\\
\refstepcounter{tbn}\thetbn\label{Qu:5}  & $x|x|^{\frac{\alpha}{1+\alpha}}$ & $c_1|x|^{\frac{\alpha}{1+\alpha}}$ & $D(1)$, $D(t)+\alpha S^1$\\
\refstepcounter{tbn}\thetbn\label{Qu:6}  & $c_2x$                           & $\ln|x|$                           & $D(1)$, $D(t)+S^0$\\
\refstepcounter{tbn}\thetbn\label{Qu:7}  & ${\rm e}^x$                      & $c_1{\rm e}^x$                     & $D(1)$, $D(t)-P(1)-S^1$\\
\refstepcounter{tbn}\thetbn\label{Qu:4}  & $c_2 x\sqrt{|x|}$                & $c_1\sqrt{|x|}+t$                  & $D(1)+S^0$, $D(t)+S^1$\\ \hline
\end{tabular}
\end{center}
\footnotesize{
The functions $\phi$ and~$\psi$ are arbitrary sufficiently smooth functions of their argument with $\phi\ne0$ and $\psi_{xx}\ne0$;\\
$\alpha$, $c_1$ and~$c_2$ are arbitrary constants with $\alpha\notin\{-1,0\}$, $c_1\ne0$ and $c_2\ne0$.
In Case~\ref{Qu:4}, $c_2>0\bmod G^\sim_1$.
}
\end{table}

\medskip\par\noindent
$\mathbf {k_2=1.}$
An appropriate subalgebra is spanned by a vector field~$Q^1=D(\tau^1)+P(\chi^1)+\alpha_1S^1+\beta_1S^0$,
where we can set~$\tau^1=1$ and $\chi^1=0$ modulo $\pi_*\hat G^\sim_1$-equivalence.
If $\alpha_1\neq0$, we can additionally set~$\beta_1=0\bmod\pi_*\hat G^\sim_1$.
Simultaneously scaling the basis element and~$t$ (if necessary), we obtain Cases~\ref{Qu:2}, \ref{Qu:2b} and~\ref{Qu:3}.

\medskip\par\noindent
$\mathbf {k_2=2.}$
Given two vector fields $Q^i=D(\tau^i)+P(\chi^i)+\alpha_iS^1+\beta_iS^0$, $i=1,2$, we can set~$\tau^1=1$ and~$\tau^2=t$
in view of Proposition~\ref{Qu:theorem:PropertiesAppropriateSub} and $\chi^1=0\bmod\pi_*\hat G^\sim_1$.
The closedness under the Lie bracket implies $[Q^1,Q^2]=Q^1$, from which $\alpha_1=0$, $\chi^2_t=0$
and~$\beta_1(\alpha_2-1)=0$. The last equation yields that either $\beta_1=0$
or $\beta_1\ne0$, $\alpha_2=1$ and hence $\beta_1=1\bmod\pi_*\hat G^\sim_1$.
The latter possibility gives Case~\ref{Qu:4},
while the former one leads to three distinct integration cases of the classifying equations~\eqref{Qu:ClassEqs},
depending on the value of~$\alpha_2$.
Cases~\ref{Qu:5}, \ref{Qu:6} and~\ref{Qu:7}
are associated with $\alpha_2\notin\{-1,0\}$, $\alpha_2=0$ and $\alpha_2=-1$, respectively.
More specifically, $\chi^2=\beta_2=0\bmod\pi_*\hat G^\sim_1$ if $\alpha_2\notin\{-1,0\}$;
$\chi^2=0\bmod\pi_*\hat G^\sim_1$, $\beta_2\neq0$ and thus $\beta_2=1\bmod\pi_*\hat G^\sim_1$ if $\alpha_2=0$;
$\chi_2\neq0$ and thus $\chi^2=-1\bmod\pi_*\hat G^\sim_1$, $\beta_2=0\bmod\pi_*\hat G^\sim_1$ if $\alpha_2=-1$.

\medskip

Each case of Lie symmetry extensions in the class~$\hat{\mathcal L}_1$ corresponds to a subclass of~$\hat{\mathcal L}_1$
parameterized by either constants or smooth functions.
Since the class~$\hat{\mathcal L}_1$ is normalized,
the equivalence group of any subclass~$\mathcal K$ of~$\hat{\mathcal L}_1$
is the subgroup of~$\hat G^\sim_1$ that consists of elements of~$\hat G^\sim_1$ preserving~$\mathcal K$.
Therefore, such elements of~$\hat G^\sim_1$ may be used for gauging parameters in arbitrary elements of~$\mathcal K$.
If the subclass~$\mathcal K$ is associated with a case of Lie symmetry extension,
then preserving~$\mathcal K$ is equivalent to preserving the corresponding appropriate subalgebra~$\mathfrak s$ of~$\hat{\mathfrak g}^\sim_1$.
In other words, the equivalence group of~$\mathcal K$ coincides with the stabilizer of~$\mathfrak s$ under the action of~$\hat G^\sim_1$.
For details see Appendix~\ref{Qu:App},
where we provide a theoretical background of the gauging procedure in detail.

Below, for each appropriate subalgebra~$\mathfrak a$ of~$\hat{\mathfrak g}$ listed in Table~\ref{Qu:tab:CompleteGroupClassificationBurgersEquations}
we present
the general form of the arbitrary elements~$A^1$ and~$A^2$ of equations invariant with respect to~$\mathfrak a$,
the constraints on the parameters of the equivalence group~$\hat G^\sim_1$
singling out the stabilizer of the corresponding appropriate subalgebra of~$\hat{\mathfrak g}^\sim_1$,
and the transformations for parameters remaining in the arbitrary elements.
Likewise Table~\ref{Qu:tab:CompleteGroupClassificationBurgersEquations}, here $\phi$ and~$\psi$
are arbitrary sufficiently smooth functions of their arguments with $\phi\ne0$, $\psi_{xx}\ne0$;
$\alpha$, $c_1$ and~$c_2$ are arbitrary constants with $\alpha\notin\{-1,0\}$ and $c_1c_2\ne0$.
\looseness=-1

\begin{enumerate}
\item $A^2=\phi(x)$, $A^1=\psi(x)$; \ $T_{tt}=X^0_t=0$; \ $\tilde \phi=(U^1)^2T_t\phi$, $\tilde \psi=U^1\psi+U^0$.
\item $A^2=\phi(x)$, $A^1=\psi(x)+t$; \ $T_t=U^1$, $X^0_t=0$; \ $\tilde \phi=(U^1)^3\phi$, $\tilde \psi=U^1\psi+U^0+U^1t-T$.
\item $A^2={\rm e}^{-2t}\phi(x{\rm e}^t)$, $A^1={\rm e}^{-t}\psi(x{\rm e}^t)$; \ $T_t=1$, $X^0=U^0=0$; \ $\tilde \phi=(U^1{\rm e}^{T-t})^2\phi$, $\tilde \psi=U^1{\rm e}^{T-t}\psi$.
\item $A^2=c_2x|x|^{\frac\alpha{1+\alpha}}$, $A^1=c_1|x|^{\frac\alpha{1+\alpha}}$; \ $T_{tt}=X^0=U^0=0$; \\
$\tilde c_2=U^1|T_tU^1|^{\frac{-\alpha}{1+\alpha}}c_2$,
$\tilde c_1=U^1|T_tU^1|^{\frac{-\alpha}{1+\alpha}}c_1$.
\item $A^2=c_2x$, $A^1=\ln|x|+c_1$; \ $T_{tt}=X^0=0$, $U^1=1$; \ $\tilde c_2=c_2$, $\tilde c_1=c_1-\ln|T_t|+U^0$.
\item $A^2=c_2{\rm e}^x$, $A^1=c_1{\rm e}^x$; \ $T_tU^1=1$, $X^0_t=U^0=0$; \ $\tilde c_2=U^1{\rm e}^{-X^0}c_2$, $\tilde c_1=U^1{\rm e}^{-X^0}c_1$.
\item $A^2=c_2x\sqrt{|x|}$, $A^1=c_1\sqrt{|x|}+t$; \ $T=U^1t+U^0$, $X^0=0$; \ $\tilde c_2=c_2\sgn U^1$, $\tilde c_1=c_1\sgn U^1$.
\end{enumerate}

The subclasses associated with Cases~\ref{Qu:2}--\ref{Qu:3}
are parameterized by arbitrary sufficiently smooth functions~$\phi$ and~$\psi$ of a single argument,
and these functions cannot be gauged with the finite-dimensional stabilizers
of the corresponding appropriate subalgebras of~$\hat {\mathfrak g}^\sim_1$.
Cases~\ref{Qu:5}--\ref{Qu:4} are considered similarly to each other.
Here we merely exemplify the gauging procedure with Case~\ref{Qu:4}.

The arbitrary elements of equations admitting the algebra
$\mathfrak a_7=\langle D(1)+S^0,\,D(t)+S^1\rangle$ are of the form $A^2=c_2 x\sqrt{|x|}$
and $A^1=c_1\sqrt{|x|}+t$, where $c_1$ and~$c_2$ are arbitrary nonzero constants.
The stabilizer subgroup of the corresponding subalgebra~$\mathfrak s_7=\langle \hat D(1)+\hat S^0,\,\hat D(t)+\hat S^1\rangle$
of~$\hat{\mathfrak g}^\sim_1$ in~$\hat{G}^\sim_1$ is singled out by the constraints $T=U^1t+U^0$, $X^0=0$.
The transformations for the parameters~$c_1$ and~$c_2$ easily follow,
$\tilde c_2=c_2\sgn U^1$, $\tilde c_1=c_1\sgn U^1$, where $U^1\neq0$.
This implies that we can gauge the sign of one of these parameters. We choose to gauge $c_2>0$.
\end{proof}

\subsection{Group classification of the regular subclass}

\begin{theorem}\label{thm:GroupClassificationL1}
A complete list of $\mathcal G^\sim_1$-inequivalent
Lie-symmetry extensions in the class~$\mathcal L_1$ is exhausted
by the cases given in Table~\ref{Qu:tab:CompleteGroupClassificationL1}.
\end{theorem}

\begin{proof}
To solve the group classification problem for the class~$\mathcal L_1$
we use the mapping technique~\cite{VaneevaPopovychSophocleous2009}.
The classes~$\hat{\mathcal L}_1$ and~$\mathcal L_1$
are related by the family of point transformations
\begin{gather}\label{Qu:L1hatL1transform}
t=\hat t,\quad  x=\hat X(\hat t,\hat x),\quad  u=\hat u,
\end{gather}
parameterized by a smooth function~$\hat X$ of $(\hat t,\hat x)$
that satisfies the nondegeneracy condition~$\hat X_{\hat x}\neq0$ and the Kolmogorov equation
\begin{gather}\label{Qu:eq:Kolmogorov}
\hat X_{\hat t}=\hat A^2\hat X_{\hat x \hat x}+\hat A^1\hat X_{\hat x}.
\end{gather}
Here and in what follows the hatted values correspond to the class~$\hat{{\mathcal L}}_1$ and ordinary values correspond to $\mathcal L_1$.
The arbitrary elements~$A^2$ and~$C$ of~$\mathcal L_1$ are
related to the arbitrary elements~$\hat A^1$ and~$\hat A^2$ of~$\hat{\mathcal L}_1$~by
\[
A^2=\hat X_{\hat x}^2\hat A^2,\quad C=\hat X_{\hat x}.
\]

Evidently, a search for the general solution of~\eqref{Qu:eq:Kolmogorov} is out of question, and therefore we
cannot solve the group classification problem for~$\mathcal L_1$ with respect to its equivalence group.
Instead, we carry out the classification of the class~$\mathcal L_1$ with respect to its equivalence groupoid~$\mathcal G_1$.
To obtain the latter classification, for each of the cases listed in Table~\ref{Qu:tab:CompleteGroupClassificationBurgersEquations},
we should take a particular solution~$\hat X$ of the Kolmogorov equation~\eqref{Qu:eq:Kolmogorov}
with the corresponding coefficients~$\hat A^1$ and~$\hat A^2$,
to compute the corresponding values of~$A^2$ and~$C$
and push forward the associated invariance algebra by the point transformation~\eqref{Qu:L1hatL1transform}
with the chosen value of~$\hat X$.
We proceed successively through all the eight cases of Lie-symmetry extensions of the class~$\hat{\mathcal L}_1$.

\begin{table}[!t]
\caption{Complete group classification of the class~${\mathcal L}_1$ up to $\mathcal G^\sim_1$-equivalence.
\label{Qu:tab:CompleteGroupClassificationL1}}
\begin{center}\newcounter{tbnn}\setcounter{tbnn}{-1}
\def\arraystretch{1.5} 
\begin{tabular}{|c|l|l|l|}
\hline
no. & \hfil $A^2$ &\hfil $C$& \hfil Basis of the invariance algebra \\
\hline
\refstepcounter{tbnn}\thetbnn\label{QuQu:1} &$A^2(t,x)$                                    & $C(t,x)$             & ---\\
\refstepcounter{tbnn}\thetbnn\label{QuQu:2} &$\phi(x)$                                     & $\psi(x)$            & $\p_t$\\
\refstepcounter{tbnn}\thetbnn\label{QuQu:2b}&$\phi(\rho)\rho_x^{-2}$                       & $\rho_x^{-1}$        & $\p_t-\rho_t\rho_x^{-1}\p_x+\p_u$\\
\refstepcounter{tbnn}\thetbnn\label{QuQu:3} &$\phi(x)$                                     & ${\rm e}^t\psi(x)$   & $\p_t-u\p_u$\\
\refstepcounter{tbnn}\thetbnn\label{QuQu:5} &$|x|^\nu$                                     & $|x|^\mu$            & $\p_t$, $(2-\nu)t\p_t+x\p_x+(\nu-\mu-1)u\p_u$\\
\refstepcounter{tbnn}\thetbnn\label{QuQu:7} &${\rm e}^{\gamma x}$                          & ${\rm e}^x$          & $\p_t$, $\gamma t\p_t-\p_x-(\gamma-1)u\p_u$\\
\refstepcounter{tbnn}\thetbnn\label{QuQu:6} &$\alpha{\zeta}\zeta_x^{-2}$                   & $\zeta_x^{-1}$       & $\p_t$, $t\p_t+\zeta\zeta_x^{-1}\p_x+\p_u$\\
\refstepcounter{tbnn}\thetbnn\label{QuQu:4} &$\alpha t^{-1}\zeta|\zeta|^{1/2}\zeta_x^{-2}$ & $t^{-2}\zeta_x^{-1}$ & $\p_t+\p_u,\ t\p_t+2\zeta\zeta_x^{-1}\p_x+u\p_u$\\
\hline
\end{tabular}
\end{center}
\footnotesize{
The functions~$\phi$ and~$\psi$ are arbitrary smooth nonvanishing functions of their arguments
such that the corresponding~$A^2$ and~$C$ satisfy the condition~\eqref{Qu:SplittingCondition} for the class~$\mathcal L_1$;
the function~$\rho$ is a particular solution of the equation $\rho_t=\phi(\rho)\rho_x^{-2}\rho_{xx}-\psi(\rho)-t$ with $\rho_x\neq0$;
the function~$\zeta$ is a particular solution of the equation $\zeta_x=\exp\big(\frac12\alpha^{-1}\ln^2|\zeta|\big)$ in Case~\ref{QuQu:6}
and of the equation $\zeta_x=\exp\big(\alpha^{-1}(4\varepsilon'|\zeta|^{1/2}+\beta\ln|\zeta|-2|\zeta|^{-1/2})\big)$
with $\varepsilon'=\sgn \zeta$ in Case~\ref{QuQu:4};
$\alpha$, $\beta$, $\gamma$, $\nu$ and $\mu$ are arbitrary constants with
$\mu(\nu-2)(\nu-\mu-1)\ne0$, $\gamma(\gamma-1)\ne0$ and $\alpha\beta\ne0$.
}
\end{table}

\medskip

\noindent{\bf0.}\
This case corresponds to the general equation in the class~$\mathcal L^1$
with the zero maximal Lie invariance algebra.

\medskip

\noindent{\bf1.}\
The equation~\eqref{Qu:eq:Kolmogorov} with time-independent $\hat A^1$ and~$\hat A^2$ admits the Lie symmetry~\smash{$\p_{\hat t}$}.
The ansatz $\hat X=\hat\vartheta(\omega)$ with $\omega=x$ and $\hat\vartheta_\omega\neq0$ for stationary solutions
of the equation~\eqref{Qu:eq:Kolmogorov} reduces it to the equation
$\hat\phi(\omega)\hat\vartheta_{\omega\omega}+\hat \psi(\omega)\hat\vartheta_{\omega}=0$.
Assuming $\hat\phi$ and~$\hat\vartheta$ varying, we consider the latter equation as an equation on~$\hat\psi$.
Then it is readily seen that the target arbitrary elements~$C$ and~$A^2$ are also time-independent,
$A^2=\psi(x)$ and $C=\phi(x)$ with $(\psi(\phi(1/\psi)_x)_x)_x\neq0$
and the pushforward of the vector field~$D(1)$ is obvious.

\medskip

\noindent{\bf2.}\
Since $\hat X_{\hat x}\ne0$, making the generalized hodograph transformation with
$(t,x)=(\hat t,\hat X)$ and $\rho=\hat x$ being the new independent and dependent variables, respectively,
we reduce the Kolmogorov equation with $\hat A^2=\hat\phi(\hat x)$ and $\hat A^1=\hat\psi(\hat x)+\hat t$
to the equation 
\[\rho_t=\hat\phi(\rho)\rho_x^{-2}\rho_{xx}-\hat\psi(\rho)-t\] 
for $\rho=\rho(t,x)$.
For each value of the parameter-function tuple $(\hat\phi,\hat\psi)$,
we should take only a single particular solution of the latter equation.
Then $C=\rho_x^{-1}$, $A^2=\hat\phi(\rho)\rho_x^{-2}$,
and the vector field $D(1)+S^0$ is pushed forward to $\p_t-\rho_t\rho_x^{-1}\p_x+\p_u$.

\newpage

\noindent{\bf3.}\
The equation~\eqref{Qu:eq:Kolmogorov}
with $\hat A^2={\rm e}^{-2\hat t}\hat\phi({\rm e}^{\hat t}\hat x)$ and $\hat A^1={\rm e}^{-\hat t}\hat\psi({\rm e}^{\hat t}\hat x)$
admits the nontrivial Lie symmetry~$\p_{\hat t}-\hat x\p_{\hat x}$,
which suggests to make the ansatz $\hat X=\hat\vartheta(\omega)$ with \smash{$\omega={\rm e}^{\hat t}\hat x$} and $\hat\vartheta_\omega\neq0$.
This ansatz reduces the above equation
to $\hat\phi(\omega)\hat\vartheta_{\omega\omega}=(\omega-\hat\psi(\omega))\hat\vartheta_{\omega}$.
We can assume $\hat\phi$ and~$\hat\vartheta$ to be arbitrary smooth functions of~$\omega$ with $\hat\phi\hat\vartheta_\omega\neq0$
and consider the reduced equation as an equation with respect to~$\hat\psi$.
Then the arbitrary elements of the class~$\mathcal L_1$ are expressed as
$A^2=\phi(x)$ and~$C={\rm e}^t\psi(x)$, where $\phi$ and~$\psi$ are arbitrary smooth nonvanishing functions of~$x$
with $(\psi(\phi(1/\psi)_x)_x)_x\neq-(1/\psi)_x$.
The last inequality follows from the condition~\eqref{Qu:SplittingCondition}
for the class~$\mathcal L_1$. The associated Lie invariance algebra is~$\langle\p_t-u\p_u\rangle$.

\medskip

\noindent{\bf4.}\
Depending on the value of the parameter~$c_1$ in the expression for~$\hat A^1$,
the corresponding Kolmogorov equation~\eqref{Qu:eq:Kolmogorov} has the particular solution
$\hat X=|\hat x|^{1-c_1}$ if $c_1\neq1$ or $\hat X=\ln|\hat x|$ otherwise.
This leads to the splitting into Case~\ref{QuQu:5} with $\mu\neq1$ and Case~\ref{QuQu:7}.
The expressions for the target arbitrary elements and the vector fields forming the maximal Lie invariance algebra easily follow
after re-denoting parameters and simplifying the transformed arbitrary elements by equivalence transformations of~$\mathcal L_1$.
The new parameters are expressed via the old ones as
$\mu=-c_1/(1-c_1)$, $\nu=\big(1-2c_1+\alpha/(\alpha+1)\big)/(1-c_1)$ if $c_1\ne1$
and $\gamma=1/(\alpha+1)$ if $c_1$.
The inequalities for the new parameters $(\mu,\nu)$ and~$\gamma$ result from their counterparts for
the old parameters $c_1$ and~$\alpha$.
These inequalities can also be derived by substituting the corresponding expressions
for~$A^2$ and~$A^1$ into the condition~\eqref{Qu:SplittingCondition} for~$\mathcal L_1$.
\medskip

\noindent{\bf5.}\
The equation~\eqref{Qu:eq:Kolmogorov} with $\hat A^2=c_2\hat x$ and $\hat A^1=\ln|\hat x|$ admits a Lie symmetry~$\p_{\tilde t}$
and thus we make the ansatz $\hat X=\vartheta(\hat x)$ to reduce this equation
to the ordinary differential equation $c_2\hat x\vartheta_{\hat x\hat x}+(\ln|\hat x|)\vartheta_{\hat x}=0$.
This equation can be easily integrated once to find $\vartheta_{\hat x}=\exp\big({-}\frac12c_2^{-1}\ln^2|\hat x|\big)$,
where the integration constant is set to zero since we look for a particular solution.
Let~$\zeta$ be the inverse of the function~$\vartheta$, $\hat x=\zeta(x)$ and thus $\zeta_x=1/\vartheta_{\hat x}$,
i.e., $\zeta$ satisfies the equation $\zeta_x=\exp\big(\frac12c_2^{-1}\ln^2|\zeta|\big)$.
To obtain $\mathcal G^\sim_1$-inequivalent cases of Lie-symmetry extensions,
we need to take a single particular solution of the last equation for each $c_2\ne0$.
The arbitrary elements of the target equation are $C=\zeta_x^{-1}$ and $A^2=c_2\zeta\zeta_x^{-2}$,
which clearly satisfy the condition~\eqref{Qu:SplittingCondition} for~$\mathcal L_1$.
The corresponding maximal Lie invariance algebra is
spanned by the vector fields~$\p_t$ and~$t\p_t+\zeta\zeta_x^{-1}\p_x+\p_u$,
which results in Case~\ref{QuQu:6}.
One can show that the function~$\xi:=\zeta\zeta_x^{-1}$ satisfies the equation $(\xi\xi_{xx})_x=0$.

\medskip

\noindent{\bf6.}\
In this case, the Kolmogorov equation~\eqref{Qu:eq:Kolmogorov} has the stationary particular solution~$\hat X={\rm e}^{-c_1\hat x}$,
which leads to the following expressions for the arbitrary elements~$A^2$ and~$C$: $A^2= c_1^2|x|^{2-1/c_1}$, $C=-c_1x$.
Denoting $\nu:=2-1/c_1$ and acting by equivalence transformations of scaling variables and of alternating their signs,
we can simplify these expressions to $A^2=|x|^{\nu}$, $C=x$ with $\nu\neq2$.
The corresponding maximal Lie invariance algebra is
$\langle\p_t,\ (2-\nu)t\p_t+x\p_x-({2-\nu})u\p_u\rangle$, which is the subcase of Case~\ref{QuQu:5} with~$\mu=1$.
\looseness=-1

\medskip

\noindent{\bf7.}\
The equation~\eqref{Qu:eq:Kolmogorov} with $\hat A^2=c_2\hat x\sqrt{|\hat x|}$ and $\hat A^1=c_1\sqrt{|\hat x|}+\hat t$
is invariant with respect to the vector field $\hat t\p_{\hat t}+2\hat x\p_{\hat x}$.
This allows us to carry out the Lie reduction of this equation by the ansatz $\hat X=\vartheta(\omega)$ with $\omega=\hat x/t^2$
to the equation
$-2\omega\vartheta_\omega=\varepsilon c_2\omega|\omega|^{1/2}\vartheta_{\omega\omega}+(\varepsilon c_1|\omega|^{1/2}+1)\vartheta_\omega$,
where $\tilde c_1:=\varepsilon c_1$, $\tilde c_2:=\varepsilon c_2$, $\varepsilon:=\sgn t$ and $\varepsilon':=\sgn \omega$.
In an attempt to find a particular solution of the reduced equation, we can integrate it once to
$\vartheta_\omega=\exp\big(-\tilde c_2^{\,-1}(4\varepsilon'|\omega|^{1/2}+\tilde c_1\ln|\omega|-2|\omega|^{-1/2})\big)$,
where the integration constant is again set to zero.
Let $\zeta$ be the inverse of the function~$\vartheta$, $\omega=\zeta(x)$ and hence $\hat x=t^2\zeta(x)$ and $\varepsilon'=\sgn \zeta$.
Taking a particular solution of the equation $\zeta_x=\exp\big(\tilde c_2^{\,-1}(4\varepsilon'|\zeta|^{1/2}+\tilde c_1\ln|\zeta|-2|\zeta|^{-1/2})\big)$
for each value of the parameter tuple $(\tilde c_1,\tilde c_2)$ with \mbox{$\tilde c_1\tilde c_2\neq0$},
we get Case~\ref{QuQu:4} with the arbitrary elements
$C={t^{-2}\zeta_x^{-1}}$ and $A^2=\tilde c_2t^{-1}\zeta|\zeta|^{1/2}\zeta_x^{-2}$ satisfying the condition~\eqref{Qu:SplittingCondition}
for~$\mathcal L_1$, and the Lie invariance algebra is $\langle\p_t+\p_u,\ t\p_t+2\zeta\zeta_x^{-1}\p_x+u\p_u\rangle$.
\end{proof}

\section{Conclusion}\label{Qu:sec:Conclusion}

The class~$\mathcal L$ of variable-coefficient Burgers equations is highly complicated
from the point of view of classifying admissible transformations and Lie symmetries.
Although its superclass~$\mathcal B$ of general Burgers equations is normalized in the usual sense
and admits several normalization-preserving gaugings of the arbitrary elements by families of equivalence transformations,
none of these properties is inherited by the class~$\mathcal L$.
This suggests to map the class~$\mathcal L$ by a family of equivalence transformations of~$\mathcal B$ to another
subclass of~$\mathcal B$ with better normalization properties.
Taking into account the experience with the class~$\mathcal B$, it is reasonable to work with a subclass of~$\mathcal B$ singled out by the constraint~$C=1$.
This is why we studied, instead of~$\mathcal L$, the class~$\hat{\mathcal L}$, which is weakly similar to~$\mathcal L$.
Although the class~$\hat{\mathcal L}$ has the same number of arbitrary elements as~$\mathcal L$ has,
the replacement of~$\mathcal L$ by~$\hat{\mathcal L}$ pays off 
since the structure of the equivalence groupoid~$\hat{\mathcal G}^\sim$ of~$\hat{\mathcal L}$
is much more comprehensible, and this is exactly what we applied the mapping of~$\mathcal L$ onto~$\hat{\mathcal L}$ for.
Indeed, the class~$\hat{\mathcal L}$ readily splits into the two $\hat{\mathcal G}^\sim$-invariant
subclasses~$\hat{\mathcal L_0}$ and~$\hat{\mathcal L}_1$
singled out from~$\hat{\mathcal L}$ by the quite simple constraints~$A^1_{xx}=0$ and $A^1_{xx}\ne0$,
respectively, which are clearly $\hat{\mathcal G}^\sim$-invariant.
In other words, the equivalence groupoid~$\hat{\mathcal G}^\sim$ is the disjoint union
of the equivalence groupoids of the subclasses~$\hat{\mathcal L}_0$ and~$\hat{\mathcal L}_1$.
The simplification of the splitting constraints in comparison with their counterparts in the class~$\mathcal L$,
expressed via a complicated differential function of both the arbitrary elements, is essential.

While the subclass~$\hat{\mathcal L}_1$ with~$A^1_{xx}\neq0$
is immediately normalized in the usual sense and thus amenable with the standard tools of the algebraic method of group classification,
there is an unexpected complication in the other branch of the general classification.
The class~$\hat{\mathcal L}_0$ is normalized in the extended generalized sense,
which is one of the main findings of this paper.
This property had been found for a number of classes of differential equations
(see, e.g., \cite{vane2015d,VaneevaPopovychSophocleous2014,VaneevaPosta2017})
but was rigorously verified for only a few of them
\cite{BoykoPopovychShapoval2015,Opanasenko2019,OpanasenkoBihloPopovych2017,OpanasenkoBoykoPopovych2019},
and in this paper we presented one example of such a verification.
We have reparameterized the class~$\hat{\mathcal L}_0$ to the class~$\bar{\mathcal L}_0$
via constructing a covering for the auxiliary system for the arbitrary elements of the former class
and have shown that the reparameterized class~$\bar{\mathcal L}_0$
is normalized in the generalized sense.
Moreover, we have found a nontrivial effective generalized equivalence group~$\breve G^\sim_0$ of~$\bar{\mathcal L}_0$,
which is a proper but not normal subgroup of the generalized equivalence group of~$\bar{\mathcal L}_0$.
Therefore, the class~$\bar{\mathcal L}_0$ admits an infinite family of effective generalized equivalence groups.
Other classes with multiple nontrivial effective generalized equivalence groups
were found in~\cite{OpanasenkoBihloPopovych2017,OpanasenkoBoykoPopovych2019}.
\looseness=1

To carry out the group classification of the class~$\mathcal L$, we combined several techniques.
As was discussed earlier, we started with the convenient splitting of the weakly similar class~$\hat{\mathcal L}$
into the normalized subclasses~$\hat{\mathcal L}_0$ and~$\hat{\mathcal L}_1$.
The class~$\hat{\mathcal L}_0$ turned out to be weakly similar to the class~${\mathcal L}_{0'}$,
which is normalized in the usual sense and which is a subclass of the class~$\mathcal L$ singled out by the constraint $C=1$.
Moreover, the group classification of the class~${\mathcal L}_{0'}$ had already been done in~\cite{PocheketaPopovych2017}.
This is why  we needed neither to carry out the group classification of the class~$\hat{\mathcal L}_0$
nor to map the classification cases.
The classification problem for the class~$\hat{\mathcal L}_1$ was tackled using the algebraic method.
Applying the mapping technique~\cite{VaneevaPopovychSophocleous2009},
we carried out the group classification of~$\mathcal L_1$ with respect to its equivalence groupoid.
It is impossible to classify Lie symmetries of equations from the class~$\hat{\mathcal L}_1$
up to $\hat G^\sim_1$-equivalence in view of arising partial differential equations
whose general solutions could not be found.

Finally, the list of $\mathcal G^\sim$-inequivalent Lie-symmetry extensions in the class~$\mathcal L$ is a disjoint
union of those of~$\mathcal L_0$ (or equivalently~$\mathcal L_{0'}$) and~$\mathcal L_1$,
which are presented in Table~1 of~\cite{PocheketaPopovych2017} and in Table~\ref{Qu:tab:CompleteGroupClassificationL1},
respectively.
Formally extending the ranges of parameters values in classification cases for the class~$\mathcal L_1$,
we can merge some classification cases for~$\mathcal L_0$ with those for~$\mathcal L_1$.
For instance, dropping the constraint $\mu\neq0$ in Case~\ref{QuQu:5} of Table~\ref{Qu:tab:CompleteGroupClassificationL1},
we combine this case with the cases of Lie-symmetry extensions in the class~$\mathcal L_0$
that are associated with the appropriate subalgebras $\mathfrak g^{2.2}=\langle \p_t, t\p_t+u\p_u\rangle$
and $\mathfrak g^{2.6}_a=\langle \p_t,t\p_t +a x\p_x+(a-1)u\p_u\rangle$ with $a\notin\{0,1/2\}$,
where $\nu=2$ and $\nu\neq2$, respectively.
Similarly, the Lie-symmetry extensions in the class~$\mathcal L_0$
with the appropriate subalgebra $\mathfrak g^{2.5}=\langle \p_t, t\p_t+\p_x-u\p_u\rangle$
can be attached to Case~\ref{QuQu:7} of Table~\ref{Qu:tab:CompleteGroupClassificationL1}.
At the same time, the equivalence groupoids of the subclasses~$\mathcal L_0$ and~$\mathcal L_1$ are
of essentially different structure, and therefore such a merge may be misleading.

Corollary~\ref{Qu:ReductionOfEqsFromLToBurgersEq} and similar assertions on equivalence of equations in the class~$\mathcal L$
can be used for finding solutions, generalized symmetries, conservation laws and other objects
related to equations that can be reduced by point transformations to well-studied equations,
like Burgers equation.

The results of the present paper show 
the relevance of the mapping and the splitting methods to the study of the class~$\mathcal L$ 
within the framework of group analysis of differential equations.

\appendix

\section{Gauging subclass parameters by equivalence transformations}\label{Qu:App}

When classifying the cases of Lie symmetry extensions for a class~$\mathcal L$ of systems of differential equations by the algebraic method,
one ends up with a collection of subclasses of~$\mathcal L$ that is associated with a complete list of inequivalent appropriate subalgebras 
of the equivalence algebra of~$\mathcal L$. 
More specifically, each of these subclasses consists of the systems in~$\mathcal L$ that admit, as their common Lie invariance algebra,  
the projection of the corresponding subalgebra in the list to the space coordinatized by the system variables. 
However, using the equivalence transformations of the class one may further gauge parameters involved in systems of such a subclass.
In particular, this gauging procedure has been used in the proof of Theorem~\ref{thm:GroupClassificationGenBurgersKdVEqs}.
Here we describe its theoretical foundations in detail.

In the notation of the first paragraph of Section~\ref{Qu:sec:EquivGroup}, let
$G^\sim_{\mathcal L}$, $\mathfrak g^\sim_{\mathcal L}$ and $\mathfrak g_\theta$ with $\theta\in\mathcal S$
be the usual equivalence group and the usual equivalence algebra of the class~$\mathcal L$
and the maximal Lie invariance algebra of the system~$\mathcal L_\theta$, respectively.
The equivalence group~$G^\sim_{\mathcal L}$ acts on its Lie algebra~$\mathfrak g^\sim_{\mathcal L}$
via pushforwards of vector fields from~$\mathfrak g^\sim_{\mathcal L}$ by transformations from~$G^\sim_{\mathcal L}$.
This action induces the action of~$G^\sim_{\mathcal L}$ on the set of subalgebras of~$\mathfrak g^\sim_{\mathcal L}$.
Elements of $G^\sim_{\mathcal L}$ and $\mathfrak g^\sim_{\mathcal L}$ can be pushed forward
by the projection $\pi$ from the space with coordinates $(x,u^{(\EqOrd)},\theta)$ to the space  with coordinates $(x,u)$.
For a subalgebra~$\mathfrak s$ of~$\mathfrak g^\sim_{\mathcal L}$, we define
the subset $\mathcal S^{\mathfrak s}:=\{\theta\in \mathcal S\mid \pi_*\mathfrak s\subseteq\mathfrak g_\theta\}$ of~$\mathcal S$
and the subclass $\mathcal L^{\mathfrak s}:=\{\mathcal L_\theta\mid \theta\in \mathcal S^{\mathfrak s}\}$ of~$\mathcal L$.
Let $\bar{\mathfrak s}$ be the maximal subalgebra of~$\mathfrak g^\sim_{\mathcal L}$ among those subalgebras
$\mathfrak h\subseteq\mathfrak g^\sim_{\mathcal L}$
such that $\pi_*\mathfrak h\subseteq \mathfrak g_\theta$ for any $\theta\in \mathcal S^{\mathfrak s}$.
(In fact, the subalgebra $\bar{\mathfrak s}$ is the span of these subalgebras and shares their defining property.)
Hence $\mathcal S^{\mathfrak s}=\mathcal S^{\bar{\mathfrak s}}$ and $\mathcal L^{\mathfrak s}=\mathcal L^{\bar{\mathfrak s}}$.
Note that $\bar{\mathfrak s}=\mathfrak s$ if $\mathfrak s$ is an appropriate algebra of~$\mathfrak g^\sim_{\mathcal L}$,
i.e., if $\pi_*\mathfrak s=\mathfrak g_\theta$ for some $\theta\in \mathcal S$.

\begin{proposition}\label{Qu:propStabilizer}
$G^\sim_{\mathcal L}\cap G^\sim_{\mathcal L^{\mathfrak s}}=\mathrm {St}_{G^\sim_{\mathcal L}}(\bar{\mathfrak s})$,
where $\mathrm {St}_{G^\sim_{\mathcal L}}(\bar{\mathfrak s})$ is the stabilizer subgroup (or, the isotropy subgroup)
of $G^\sim_{\mathcal L}$ with respect to~$\bar{\mathfrak s}$.
\end{proposition}

\begin{proof}
Let us prove the inclusion $G^\sim_{\mathcal L}\cap G^\sim_{\mathcal L^{\mathfrak s}}\supseteq\mathrm {St}_{G^\sim_{\mathcal L}}(\bar{\mathfrak s})$.
Suppose that $\mathcal T\in\mathrm {St}_{G^\sim_{\mathcal L}}(\bar{\mathfrak s})$.
Then $\mathcal T\in G^\sim_{\mathcal L}$ and $\mathcal T_*\bar{\mathfrak s}=\bar{\mathfrak s}$.
The former condition implies that for any tuple $\theta\in\mathcal S^{\mathfrak s}\subseteq\mathcal S$
we have $\mathcal T\theta\in\mathcal S$ (where $\mathcal T$ acts on~$\theta$ as on a tuple of functions)
and $\mathfrak g_{\mathcal T\theta}=(\pi_*\mathcal T)_*\mathfrak g_\theta$.
On the other hand, $(\pi_*\mathcal T)_*\mathfrak g_\theta\supseteq(\pi_*\mathcal T)_*(\pi_*\bar{\mathfrak s})=\pi_*(\mathcal T_*\bar{\mathfrak s})=\pi_*\bar{\mathfrak s}$,
yielding that $\mathfrak g_{\mathcal T\theta}\supseteq\pi_*\bar{\mathfrak s}$, i.e.,
$\mathcal T\theta\in\mathcal S^{\mathfrak s}$ for any $\theta\in\mathcal S^{\mathfrak s}$.
Therefore, $\mathcal T\in G^\sim_{\mathcal L^{\mathfrak s}}$.

Now we prove the inclusion $G^\sim_{\mathcal L}\cap G^\sim_{\mathcal L^{\mathfrak s}}\subseteq\mathrm {St}_{G^\sim_{\mathcal L}}(\bar{\mathfrak s})$.
If $\mathcal T\in G^\sim_{\mathcal L}\cap G^\sim_{\mathcal L^{\mathfrak s}}$,
then $\mathcal T^{-1}\in G^\sim_{\mathcal L}\cap G^\sim_{\mathcal L^{\mathfrak s}}$ as well, and
$\mathcal T\mathcal S^{\mathfrak s}=\mathcal T^{-1}\mathcal S^{\mathfrak s}=\mathcal S^{\mathfrak s}$.
Fix an arbitrary $\theta\in\mathcal S^{\mathfrak s}$.
Then $\hat\theta:=\mathcal T^{-1}\theta\in\mathcal S^{\mathfrak s}$,
and thus $\mathfrak g_{\hat\theta}\supseteq\pi_*\bar{\mathfrak s}$.
We obtain
\[
\mathfrak g_\theta=\mathfrak g_{\mathcal T\hat\theta}
=(\pi_*\mathcal T)_*\mathfrak g_{\hat\theta}\supseteq(\pi_*\mathcal T)_*\pi_*\bar{\mathfrak s}
=\pi_*(\mathcal T_*\bar{\mathfrak s}).
\]
In view of the above maximality of~$\bar{\mathfrak s}$, we get $\mathcal T_*\bar{\mathfrak s}\subseteq\bar{\mathfrak s}$.
Similarly, $(\mathcal T^{-1})_*\bar{\mathfrak s}\subseteq\bar{\mathfrak s}$,
which is equivalent to $\bar{\mathfrak s}\subseteq\mathcal T_*\bar{\mathfrak s}$,
and thus the statement follows.
\end{proof}

We can also consider Lie algebras instead of Lie groups.
Thus, we look at the infinitesimal version of Proposition~\ref{Qu:propStabilizer}.
First, let us recall that the equivalence group and the equivalence algebra of a class of differential equations
may not be finite-dimensional and therefore we should actually speak of its Lie equivalence pseudogroup
and the associated Lie equivalence algebroid.
To prove the next assertion, we recall a crucial property of Lie pseudogroups and Lie algebroids.
Namely, they are defined as a set of local solutions of systems of (formally integrable) differential equations
\cite[Chapter~5.B]{Pommaret1994}.
In fact, these systems for the equivalence group and for the equivalence algebra
of a class of differential equations are simply the associated systems of determining equations.

Let $\Psi$ be a differential vector-function in components of vector fields that defines the Lie algebroid~$\bar {\mathfrak s}$,
i.e., a vector field~$\mathrm u$ belongs to~$\bar{\mathfrak s}$, $\mathrm u\in\bar{\mathfrak s}$,
if and only if $\Psi(\mathrm u)=0$.
The function $\Psi$ may also be seen as a tuple of linear differential operators
acting on the components of its argument.

\begin{corollary}\label{Qu:lemmaNormalizer}
$\mathfrak g^\sim_{\mathcal L}\cap\mathfrak g^\sim_{\mathcal L^{\mathfrak s}}=\mathrm N_{\mathfrak g^\sim_{\mathcal L}}(\bar{\mathfrak s})$,
where $\mathrm N_{\mathfrak g^\sim_{\mathcal L}}(\mathfrak s)$ is the normalizer of~$\bar{\mathfrak s}$ in $\mathfrak g^\sim_{\mathcal L}$.
\end{corollary}

\begin{proof}
Let $\mathrm v\in\mathfrak g^\sim_{\mathcal L}\cap\mathfrak g^\sim_{\mathcal L^{\mathfrak s}}$.
Then $\exp(\varepsilon\mathrm v)\in G^\sim_{\mathcal L}\cap G^\sim_{\mathcal L^{\mathfrak s}}$.
By Proposition~\ref{Qu:propStabilizer}, for any $\mathrm w\in\bar{\mathfrak s}$
we have $(\exp(\varepsilon\mathrm v))_\ast\mathrm w\in\bar{\mathfrak s}$.
Recall that
\[
[\mathrm v,\mathrm w]
=\mathscr L_{\mathrm v}\mathrm w
=\frac{\mathrm d\mathrm w_\varepsilon}{\mathrm d\varepsilon}\Big|_{\varepsilon=0}
=\lim\limits_{\varepsilon\to0}\frac{\mathrm w_\varepsilon-\mathrm w_0}\varepsilon,
\quad\mbox{where}\quad
\mathrm w_\varepsilon\big|_p:=\big(\exp(-\varepsilon\mathrm v)_\ast\mathrm w\big)\big|_{\exp(\varepsilon\mathrm v)(p)}
\]
for any point $p$ in the space coordinatized by $(x,u^{(r)},\theta)$.
The vector field~$\mathrm u_\varepsilon=(\mathrm w_\varepsilon-\mathrm w_0)/\varepsilon$
clearly belongs to~$\bar{\mathfrak s}$ for any~$\varepsilon\ne0$.
Thus, $\Psi(\mathrm u_\varepsilon)=0$ and recalling that components of $\mathrm u_\varepsilon$
are smooth in $(x,u^{(r)},\theta,\varepsilon)$,
we use the continuity of~$\Psi$ in components of vector fields to show
\[
\Psi([\mathrm v,\mathrm w])=\Psi\left(\lim\limits_{\varepsilon\to0}\mathrm u_\varepsilon\right)=
\lim\limits_{\varepsilon\to0}\Psi(\mathrm u_\varepsilon)=0,
\]
implying that $[\mathrm v,\mathrm w]\in\bar{\mathfrak s}$ and thus
$\mathrm v\in\mathrm N_{\mathfrak g^\sim_{\mathcal L}}(\bar{\mathfrak s})$.

Conversely, if $\mathrm v\in\mathrm N_{\mathfrak g^\sim_{\mathcal L}}(\bar{\mathfrak s})$,
then for any $\mathrm w\in\bar{\mathfrak s}$ and for any~$n\in\mathbb N_0$
we have $(\mathop{\rm ad}\mathrm v)^{(n)}(\mathrm w)\in\bar{\mathfrak s}$ as well.
This implies
\[
(\exp(\varepsilon \mathrm v))_\ast\mathrm w=\sum_{n=0}^\infty\frac{\varepsilon^n(\mathop{\rm ad}\mathrm v)^n(\mathrm w)}{n!}
\in\bar{\mathfrak s}
\]
due to the continuity argument as above and the convergence of the sequence of partial sums of the series.
Thus $\exp(\varepsilon\mathrm v)$ lies in the stabilizer of~$\bar{\mathfrak s}$ in~$G^\sim_{\mathcal L}$,
which is the same as the Lie pseudogroup~$G^\sim_{\mathcal L}\cap G^\sim_{\mathcal L^{\mathfrak s}}$
in view of Proposition~\ref{Qu:propStabilizer}. In turn, this implies that $\mathrm v$ belongs to
its Lie algebroid~$\mathfrak g^\sim_{\mathcal L}\cap \mathfrak g^\sim_{\mathcal L^{\mathfrak s}}$.
\end{proof}

Corollary~\ref{Qu:lemmaNormalizer} is analogous to a well-known fact in Lie theory, cf.~\cite[Chapter~11]{HilgertNeeb2012}.

\begin{corollary}
If the class~$\mathcal L$ is normalized, then
$G^\sim_{\mathcal L^{\mathfrak s}}=\mathrm {St}_{G^\sim_{\mathcal L}}(\bar{\mathfrak s})$
and $\mathfrak g^\sim_{\mathcal L^{\mathfrak s}}=\mathrm N_{\mathfrak g^\sim_{\mathcal L}}(\bar{\mathfrak s})$.
\end{corollary}

\begin{proof}
If the class~$\mathcal L$ is normalized, then $G^\sim_{\mathcal L^{\mathfrak s}}\leqslant G^\sim_{\mathcal L}$
and $\mathfrak g^\sim_{\mathcal L^{\mathfrak s}}\subseteq\mathfrak g^\sim_{\mathcal L}$,
from which the statement readily follows.
\end{proof}

To summarize the above, in order to gauge arbitrary elements of the subclass~$\mathcal L^{\mathfrak s}$ up to $G^\sim_{\mathcal L}$-equivalence
we need to know the group~$G^\sim_{\mathcal L^{\mathfrak s}}\cap G^\sim_{\mathcal L}$.
We may compute it either as a subgroup of~$G^\sim_{\mathcal L}$ that preserves the subclass~$\mathcal L^{\mathfrak s}$ or
as the stabilizer subgroup~$\mathrm{St}_{G^\sim_{\mathcal L}}(\bar{\mathfrak s})$ of~$G^\sim_{\mathcal L}$ with respect to~$\mathfrak s$.
The identity component of the stabilizer can be found
via computing the infinitesimal generators forming the Lie algebra~$\mathrm N_{\mathfrak g^\sim_{\mathcal L}}(\bar{\mathfrak s})$
first and exponentiating thereafter but then the other components of the stabilizer are missed.

However, not all these transformations are essential for the gauging procedure
because the projections of some of them under $\pi_*$ may be symmetry transformations
for all equations in~$\mathcal L^{\mathfrak s}$ and thus do not change their form.
In fact, these projections constitute the group~$G^\cap_{\mathcal L^{\mathfrak s}}\cap \pi_*G^\sim_{\mathcal L}$,
which is a normal subgroup of~$\pi_*\left(G^\sim_{\mathcal L^{\mathfrak s}}\cap G^\sim_{\mathcal L}\right)$.
Here $G^\cap_{\mathcal L^{\mathfrak s}}$ denotes the kernel Lie point symmetry group of the class~$\mathcal L^{\mathfrak s}$.
Therefore, to carry out the gauging procedure efficiently we have to factor out these transformations.

\section*{Acknowledgments}

The authors are grateful to the reviewer for useful suggestions. 
The authors also thank Michael Kunzinger, Dmytro Popovych and Galyna Popovych for helpful discussions.
SO acknowledges the support of the Natural Sciences and Engineering Research Council of Canada
and the hospitality of the University of Vienna.
The research of AB was undertaken, in part, thanks to funding from the Canada Research Chairs program and the NSERC Discovery Grant program.
The research of ROP was supported by the Austrian Science Fund (FWF), projects P25064 and P30233.

\end{document}